\documentclass[10pt,onecolumn,prd,aps,showkeys]{revtex4-2}

\RequirePackage{amsmath,amssymb,amsthm}
\RequirePackage[dvipsnames,usenames]{color}
\usepackage{bm}
\usepackage[british]{babel}
\usepackage[latin1]{inputenc}
\usepackage[T1]{fontenc}
\usepackage[dvipsnames]{xcolor}

\usepackage{pgfplots}
\pgfplotsset{compat=1.18}
\usetikzlibrary{arrows.meta}
\usepgfplotslibrary{fillbetween}

\newtheorem{theorem}{Theorem}

\newtheorem{corollary}{Corollary}

\newcommand{\be}{\begin{equation}} 
	\newcommand{\ee}{\end{equation}}

\renewcommand{\leq}{\leqslant}
\renewcommand{\geq}{\geqslant}

\newcommand{\omegap}{\omega^{\prime}}

\newcommand{\Ltre}{L_{1}}
\newcommand{\Lquattro}{L_{2}}

\numberwithin{equation}{section}

\begin{document}
\title{Six-Field Rational Extended Thermodynamics of\\
Polyatomic Gases in Curved Spacetime}

	\author{Luca Gallerani}
	\email{luca.gallerani6@unibo.it}
	\affiliation{Department of Mathematics \& $\mathcal{AM}^2$, University of Bologna, Italy} 
	\author{Andrea Giusti}
	\email{andrea.giusti9@unibo.it}
	\affiliation{DIFA \& $\mathcal{AM}^2$, University of Bologna, Italy} 
	\affiliation{INFN, Sezione di Bologna, IS FLAG, Italy} 
    \author{Andrea Mentrelli}
	\email{andrea.mentrelli@unibo.it}
	\affiliation{Department of Mathematics \& $\mathcal{AM}^2$, University of Bologna, Italy} 
    \affiliation{INFN, Sezione di Bologna, IS FLAG, Italy} 
	\author{Tommaso Ruggeri}
	\email{tommaso.ruggeri@unibo.it}
	\affiliation{Department of Mathematics \& $\mathcal{AM}^2$, University of Bologna, Italy} 
	\affiliation{Accademia Nazionale dei Lincei, Rome, Italy} 
\begin{abstract}
\noindent
We formulate a generally covariant six-field Rational Extended Thermodynamics
model (RET$_6$) for relativistic polyatomic gases, with the dynamical pressure
as the only non-equilibrium variable. The model is based on a polyatomic
extension of the Boltzmann--Chernikov kinetic equation, where the one-particle
distribution depends also on an internal-energy variable, and on the Maximum
Entropy closure of the associated relativistic moment hierarchy. The resulting
field equations, closure relations, and production term are therefore fixed by
the underlying kinetic structure rather than postulated phenomenologically. We
extend the RET$_6$ model from Minkowski spacetime to a general curved spacetime
by the minimal coupling prescription and couple it to the Einstein equations.

\noindent
As a first structural result, we prove a kinetic-theory no-go theorem in this
polyatomic RET setting stating that any stress-energy tensor induced by a non-negative
relativistic one-particle distribution function satisfies the strong energy
condition. 

\noindent
We then specialize the theory to a homogeneous and isotropic
Friedmann--Lema\^{i}tre--Robertson--Walker (FLRW) spacetime. In this setting the
dynamical pressure modifies the expansion dynamics with respect to the
perfect-fluid Euler case, but the no-go theorem excludes acceleration driven by
the RET$_6$ gas alone. Finally, we reintroduce a cosmological constant and study
the combined $\Lambda$RET$_6$ model. For the diatomic equation of state and a
constant positive relaxation time, we prove the existence and local stability
of a de~Sitter attractor at late times. Numerical integrations show that, for
representative post-recombination initial data and constant relaxation times, the
expansion history rapidly approaches that of $\Lambda$CDM, with small
non-equilibrium corrections controlled by the relaxation time and by the
initial value of the dynamical pressure.
\end{abstract}
%
%
%
%
\maketitle
\newpage
%
%
%
%
\section{Introduction}

A central question in the theory of dissipative relativistic fluids
\cite{VereshchaginAksenov2017} is how to construct field equations that are simultaneously consistent with relativistic
causality, with the second law of thermodynamics, and with an underlying
kinetic description of the gas. Phenomenological approaches, such as those
based on M\"{u}ller--Israel-type theories
\cite{Muller1966,Israel1976}, typically postulate the
form of the dissipative fluxes and their relaxation equations a priori,
fixing free parameters by physical intuition or by comparison with kinetic
theory only a posteriori. Note that we use the designation {\em M\"{u}ller--Israel} to
reflect the historical priority. Indeed, it was M{\"u}ller \cite{Muller1966} who
first formulated, in his doctoral thesis, a causal theory of classical and relativistic dissipative
fluids with relaxation-type equations. The relativistic framework was
subsequently developed independently by Israel \cite{Israel1976}.

Rational Extended Thermodynamics (RET) \cite{MuellerRuggeri1998, newbook} offers a
structurally different route. The first relativistic RET theory was
formulated by Liu, M\"{u}ller, and Ruggeri \cite{LiuMuellerRuggeri1986} using the universal
principles of continuum mechanics (i.e., the entropy principle and the
concavity of the entropy density) as closure of a hierarchy of $14$ moment
equations derived from the Boltzmann--Chernikov kinetic equation
\cite{BGK, Synge, KC}, without postulating any constitutive relation a
priori. In this approach the field equations, the closure relations for the
higher-order moments, and the production terms are all determined by the
entropy structure of the system. A second, complementary closure strategy is the Maximum Entropy Principle
(MEP), whose use in the context of moment hierarchies has its own
history. In the classical framework, the MEP was first applied to rarefied
gas dynamics by Kogan \cite{kogan}, and subsequently used systematically
by Dreyer \cite{dreyer1}, by M\"{u}ller and Ruggeri \cite{ET}, by
Levermore \cite{levermore1996}, and by Boillat and Ruggeri
\cite{boillat_ruggeri1997}.

In the relativistic context, the MEP was first applied to monatomic gases by
Dreyer \cite{dreyer2}. The first relativistic RET theory for polyatomic gases was then developed by Pennisi and Ruggeri
\cite{PennisiRuggeri2017} through a double hierarchy of moments, accounting
for both the particle four-momentum and an additional internal-energy
variable. In particular, applying the MEP to the equilibrium subsystem yields
the polyatomic  analogue of the J\"uttner equilibrium
distribution, with the internal-energy variable and its density of states
entering the equilibrium distribution.
The construction was then carried over to the new hierarchy by Arima,
Carrisi, Pennisi, and Ruggeri \cite{Arima2022} for the 15-field theory.

 A detailed
historical account of the role of the MEP in RET is given in
\cite{MEP}. In both the entropy-principle approach and the MEP approach, the resulting closure relations coincide; moreover, they agree with those obtained by Marle \cite{Marle} by means of the relativistic counterpart of Grad's classical method \cite{Grad} (see  \cite{dreyer2,MuellerRuggeri1998,Arima2022}). This agreement confirms that the macroscopic closure is not postulated independently of the microscopic kinetic description.

For the broader mathematical and physical background on relativistic fluids
and magneto-fluids, including applications to astrophysics and plasma physics
and the analysis of waves and shock waves, we refer to Anile's classical
monograph \cite{Anile1989}. The use of the relativistic Boltzmann equation on
non-flat backgrounds has also been developed in the standard kinetic setting.
In particular, Kremer
obtained relativistic field equations from the Boltzmann equation in special
and general relativity, including applications to a homogeneous and isotropic
Universe and to a Schwarzschild spacetime \cite{KremerAIP2012}. He also
studied a relativistic gas in a Schwarzschild metric, deriving the dependence
of the transport coefficients and of the relativistic Fourier law on the
gravitational field \cite{KremerJSM2013}, and later analyzed diffusion in
relativistic gas mixtures in gravitational fields \cite{KremerPhysicaA2014}.
These works concern the standard relativistic kinetic framework rather than
the polyatomic RET$_6$ closure used below, but they provide useful kinetic
precedents for applying Boltzmann-type descriptions on curved spacetimes.

In particular, relativistic polyatomic gases admit a six-field model
(RET$_6$) \cite{Arima2022} in which, besides the particle
density and the four-velocity, the energy density and a single non-equilibrium
variable --- the dynamical pressure $\Pi$ --- appear as independent fields.
As we recall in Sec.~\ref{sec:RET}, RET$_6$ is not postulated directly but
arises as a {\em principal subsystem}, in the sense of Boillat and Ruggeri
\cite{Arma}, of the more general $15$-field relativistic kinetic theory of
polyatomic gases (RET$_{15}$) \cite{Arima2022}; this guarantees that it
inherits the convex entropy law and the subcharacteristic property of the
parent theory, and reduces the independent fields to
$(\rho, T, u^\alpha, \Pi)$.

The dynamical pressure represents the non-equilibrium correction associated
with the internal degrees of freedom of the gas, such as rotational and
vibrational modes in the polyatomic case. It is the only dissipative variable
that survives in a spatially homogeneous and isotropic configuration, and
therefore RET$_6$ provides the natural minimal extension of perfect-fluid
cosmology.

The present paper has two main goals. The
first, and primary, goal is theoretical: we extend the RET$_6$ model for
relativistic polyatomic gases from Minkowski spacetime to a general curved
spacetime by minimal coupling, establishing its general-covariant field
equations within the kinetically based RET closure obtained by the Maximum
Entropy Principle. The second goal is to illustrate the resulting general
theory through a significant application of broader interest, namely the
cosmological dynamics of a homogeneous and isotropic universe sourced by a
RET$_6$ gas.

Our first result is the kinetic-theory no-go theorem, which states that any
stress-energy tensor induced by a non-negative relativistic one-particle
distribution function satisfies the strong energy condition. The connection
between relativistic kinetic matter and the standard energy conditions is, of
course, classical; it goes back to the foundations of relativistic kinetic
theory and is routinely used in the Einstein--Vlasov literature
\cite{Ehlers1971,HawkingEllis1973,Wald1984,Andreasson2011,SarbachZannias2013}. The
point of the present theorem is therefore not to claim such positivity as a new
fact for kinetic theory in general, but to state and use the corresponding
pointwise consequence explicitly for the polyatomic RET hierarchy and its
curved-spacetime uplift. In particular, this gives a clean interpretation of
RET$_6$ in cosmology; namely, that as long as the closure remains an admissible kinetic
closure, the strong energy condition is inherited by the macroscopic model,
independently of the number of moments retained in the hierarchy.

As an application of the general theory, we specialize the curved-spacetime
RET$_6$ equations to a spatially flat Friedmann--Lema\^{i}tre--Robertson--Walker
(FLRW) spacetime and derive the corresponding cosmological evolution
equations. In the polyatomic case the dynamical pressure is genuinely
non-zero even in the classical limit, in contrast with the monatomic gas,
where it appears only as a relativistic correction of order
$\mathcal{O}(c^{-4})$ \cite{MuellerRuggeri1998}, $c$ being the speed of
light; this makes the polyatomic RET$_6$ model the natural candidate for a
non-trivial non-equilibrium correction to the cosmic expansion. As an
immediate corollary of the no-go theorem, an accelerated expansion cannot be
generated by the RET$_6$ gas alone. Thus, a cosmological constant (or some other
form of dark energy beyond standard kinetic matter) remains necessary. We
study the combined $\Lambda$RET$_6$ model, both analytically, establishing
the existence and local asymptotic stability of a de~Sitter attractor at late
times, and numerically, showing that for physically motivated
post-recombination parameter values the expansion history converges rapidly to
that of $\Lambda$CDM, while exhibiting
small but appreciable initial-condition-dependent deviations gauged by the
non-equilibrium dynamical pressure.

The structure of the paper is as follows. In Sec.~II we recall the RET$_6$
model for relativistic polyatomic gases on Minkowski spacetime and its
convective form. In Sec.~III we extend the model to a general curved spacetime
by minimal coupling and specialize it to the FLRW geometry. In Sec.~IV we
establish the kinetic-theory no-go theorem and its cosmological corollary. In
Sec.~V we reintroduce the cosmological constant, study the $\Lambda$RET$_6$
model, and establish the local stability of the de~Sitter attractor. In
Sec.~VI we present numerical solutions and compare the resulting expansion
history with $\Lambda$CDM. Finally, in Sec.~VII we present some concluding
remarks.
\section{The RET$_6$ model on Minkowski spacetime \label{sec:RET}}
We briefly recall the relativistic extended thermodynamics model with six
fields (RET$_6$) for a polyatomic gas in special relativity
\cite{Arima2022}. Rather than being postulated directly,
RET$_6$ is obtained from relativistic kinetic theory as a principal subsystem
of a larger fifteen-field theory. We therefore start from the
Boltzmann--Chernikov equation, recall the associated RET$_{15}$ moment system,
and then explain how RET$_6$ follows as a principal subsystem in the sense of
Boillat and Ruggeri \cite{Arma}.

We work on Minkowski spacetime, endowed in Cartesian coordinates with the
metric
\[
    \eta_{\mu\nu}=\mathrm{diag}(+1,-1,-1,-1).
\]
The four-velocity satisfies the normalization condition
\[
    u^\alpha u_\alpha=c^2,
\]
and the projection tensor onto the local rest space orthogonal to $u^\alpha$
is
\[
    h^{\alpha\beta}:=
    \frac{u^\alpha u^\beta}{c^2}-\eta^{\alpha\beta}.
\]
The particle mass density is denoted by $\rho=mn$, where $m$ is the rest mass
of a molecule and $n$ is the particle number density. The total energy density
is written as
\begin{equation}\label{eenergia}
    e=\rho c^2+\rho\varepsilon,
\end{equation}
where $\varepsilon$ is the specific internal energy. We also introduce the
dimensionless inverse temperature
\begin{equation}
    \gamma:=\frac{mc^2}{k_{\rm B}T},
\end{equation}
where $k_{\rm B}$ denotes Boltzmann's constant.

At the kinetic level, the monatomic relativistic gas is described by the
Boltzmann--Chernikov equation \cite{BGK,Synge,KC}. For polyatomic  gases, Pennisi and Ruggeri \cite{PennisiRuggeri2017} proposed an
extension in which the distribution function depends also on an internal
energy variable,
\[
    f\equiv f(x^\alpha,p^\beta,\mathcal I).
\]
The transport part keeps the same Boltzmann--Chernikov form,
\begin{equation}\label{BoltzR}
    p^\alpha \partial_\alpha f=Q,
\end{equation}
where $p^\alpha$ is the particle four-momentum and $Q$ denotes the
collision term. Its explicit form, based on the BGK-type model for
relativistic monatomic and polyatomic gases proposed by Pennisi and
Ruggeri~\cite{PennisiRuggeri2018}, will be specified below when discussing
the production term.
The indices $\alpha,\beta$ take the values $0,1,2,3$. The variable
$\mathcal I$ accounts for rotational and vibrational molecular modes, and the
density of internal states is denoted by $\phi(\mathcal I)$; see
\cite{newbook} and references therein.

Following the new hierarchy of moments introduced in \cite{Arima2022}, the
relativistic balance equations for a polyatomic gas are
\begin{equation}
    \partial_\alpha V^\alpha=0,\qquad
    \partial_\alpha T^{\alpha\beta}=0,\qquad
    \partial_\alpha A^{\alpha\beta\gamma}=I^{\beta\gamma}.
    \label{eq:RET15_moments}
\end{equation}
The corresponding moments are defined by
\begin{align}
\begin{split}
    V^\alpha
    &= mc\int_{\mathbb R^3}\!\!\int_0^{+\infty}
       f\,p^\alpha\,\phi(\mathcal I)\,d\mathcal I\,d\boldsymbol P,\\
    T^{\alpha\beta}
    &= \frac{1}{mc}\int_{\mathbb R^3}\!\!\int_0^{+\infty}
       f\,p^\alpha p^\beta\left(mc^2+\mathcal I\right)
       \phi(\mathcal I)\,d\mathcal I\,d\boldsymbol P,\\
    A^{\alpha\beta\gamma}
    &= \frac{1}{m^3c^3}\int_{\mathbb R^3}\!\!\int_0^{+\infty}
       f\,p^\alpha p^\beta p^\gamma
       \left(mc^2+\mathcal I\right)^2
       \phi(\mathcal I)\,d\mathcal I\,d\boldsymbol P,\\
    I^{\beta\gamma}
    &= \frac{1}{m^3c^3}\int_{\mathbb R^3}\!\!\int_0^{+\infty}
       Q\,p^\beta p^\gamma
       \left(mc^2+\mathcal I\right)^2
       \phi(\mathcal I)\,d\mathcal I\,d\boldsymbol P.
\end{split}
\label{eq:RET15_moments_def}
\end{align}
Here $d\boldsymbol P=d^3p/p_0$ is the Lorentz-invariant measure in momentum
space. The essential point of this hierarchy is that the $n$-th moment
contains the $n$-th power of the total molecular energy $mc^2+\mathcal I$,
namely the sum of the rest energy and the internal energy of the molecule.

This distinguishes the present model from the earlier relativistic
polyatomic-gas model of Pennisi and Ruggeri \cite{PennisiRuggeri2017}, where
the third-order moment involved the factor $mc^2+2\mathcal I$. In the newer
hierarchy of \cite{Arima2022}, the moments are built instead from powers of
the full molecular energy, in closer analogy with the usual moment
construction.
The closure of the system \eqref{eq:RET15_moments} is obtained by means of
the Maximum Entropy Principle \cite{kogan,dreyer1,ET,MEP}, following the
relativistic polyatomic-gas construction of Arima, Carrisi, Pennisi, and
Ruggeri \cite{Arima2022}. This yields the fifteen-field theory RET$_{15}$.
The balance laws involve the moments $V^\alpha$, $T^{\alpha\beta}$, and
$A^{\alpha\beta\gamma}$, while the independent fields of the closed theory
are most conveniently represented by the usual fourteen physical variables
\[
    \rho,\quad T,\quad u^\alpha,\quad \Pi,\quad q^\alpha,\quad
    t^{\langle\alpha\beta\rangle},
\]
together with the additional scalar non-equilibrium variable $\Delta$,
associated with the trace part of the third-order moment.

In
this way, RET$_{15}$ provides a macroscopic counterpart of the relativistic
Boltzmann--Chernikov equation \eqref{BoltzR} for polyatomic gases. Its
constitutive quantities are determined in terms of a scalar function
$\omega(\gamma)$, which encodes the thermal equation of state, together with
its derivatives \cite{Arima2022}.

A key point for the present work is that important reduced theories are not
introduced by additional phenomenological assumptions. They arise as
\emph{principal subsystems} of RET$_{15}$, in the precise sense of Boillat and
Ruggeri \cite{Arma}. A principal subsystem is obtained by setting a suitable
subset of the Lagrange multipliers equal to zero, suppressing the corresponding
balance equations, and expressing the associated non-equilibrium variables in
terms of the remaining fields. This construction preserves the main structural
properties of the parent theory: the reduced system still satisfies a convex
entropy law and inherits symmetric hyperbolicity. Moreover, the
subcharacteristic property follows automatically. More precisely, Boillat and
Ruggeri proved that the characteristic velocities of a principal subsystem lie
within the characteristic cone of the full system, i.e. the maximal characteristic
velocity of the subsystem is not larger than that of the parent system, while
its minimal characteristic velocity is not smaller.

There are two relevant principal subsystems of RET$_{15}$. The first one is
RET$_{14}$, obtained by imposing that the trace of the second-order Lagrange
multiplier vanishes,
\[
    \lambda^\alpha{}_\alpha=0.
\]
It retains the particle flux, the energy-momentum tensor, and the trace-free
part of the third-order balance law. Its independent fields are
\[
    (\rho,T,u^\alpha,\Pi,q^\alpha,t^{\langle\alpha\beta\rangle}),
\]
namely fourteen fields in total.

The second one is RET$_6$, which is the model used in this paper. It is
obtained by imposing the condition
\begin{equation}\label{lambda_cond}
    \lambda_{\langle\mu\nu\rangle}\equiv
    \lambda_{\mu\nu}
    -\frac{1}{4}\lambda^\alpha{}_\alpha\,g_{\mu\nu}=0,
\end{equation}
that is, the trace-free part of the second-order Lagrange multiplier is set
to zero. This condition forces the heat flux and shear stress to vanish,
\[
    q^\alpha=0,\qquad
    t^{\langle\mu\nu\rangle}=0,
\]
and leaves only the scalar non-equilibrium variable, the dynamical pressure
$\Pi$. In addition, one obtains the algebraic relation
\begin{equation}\label{Delta6}
    \Delta^{(6)}=w\,\Pi,
    \qquad
    w=4c^2\,\frac{D^{44}_4-3D^{43}_4}{D^{34}_4-3D^{33}_4},
\end{equation}
where $D^{ij}_4$ are the minor determinants of the matrix $D_4$ defined in
\cite[Sec.~7.2]{Arima2022}. We stress that \eqref{Delta6} differs from the
corresponding relation in RET$_{14}$.

The independent fields of RET$_6$ are therefore
\[
    (\rho,T,u^\alpha,\Pi),
\]
namely six fields, and the surviving balance laws are
\begin{equation}\label{eq:trace}
    \partial_\alpha V^\alpha=0,\qquad
    \partial_\alpha T^{\alpha\beta}=0,\qquad
    \partial_\alpha A^{\alpha\beta}{}_{\beta}=I^\beta{}_{\beta}.
\end{equation}
It is important to emphasize that, unlike in the classical theory, RET$_{14}$
and RET$_6$ are at the \emph{same level} in the relativistic theory, meaning that 
both are direct principal subsystems of RET$_{15}$ \cite{Arima2022}. Thus RET$_6$ is not
obtained as a subsystem of RET$_{14}$, as it happens in the classical limit.
Rather, RET$_{14}$ and RET$_6$ are obtained directly from RET$_{15}$ by
imposing different frozen conditions on the Lagrange multipliers. This is a genuinely
relativistic feature, related to the coupling between the scalar part
$\Delta$ and the tensorial part associated with the third-order moment
$A^{\alpha\beta\gamma}$.

Physically, the dynamical pressure $\Pi$ represents the non-equilibrium
correction associated with the internal degrees of freedom of the gas, such as
rotational and vibrational modes in the polyatomic case. It is the only
dissipative variable that survives in a spatially homogeneous and isotropic
configuration, since heat flux and shear stress vanish by symmetry
\cite{Ellisetal:2012}. For this reason RET$_6$ is the natural minimal
extension of the perfect-fluid Euler theory capable of describing
bulk-viscosity-type dissipation in a cosmological setting. At the same time,
the model remains determined by the underlying relativistic kinetic theory of
polyatomic gases once the density of internal states $\phi(\mathcal I)$ and
the kinetic collision model, including the relaxation time, have been
specified.

For the reader's convenience, we now collect the explicit expressions used in
the sequel. The particle four-vector and the energy-momentum tensor are
\begin{align}\label{Tab}
\begin{split}
    V^\alpha &= \rho u^\alpha,\\
    T^{\alpha\beta}
    &= \frac{e}{c^2}u^\alpha u^\beta+(p+\Pi)h^{\alpha\beta}\\
    &= \frac{e+p+\Pi}{c^2}u^\alpha u^\beta
       -(p+\Pi)\eta^{\alpha\beta},
\end{split}
\end{align}
where $p$ is the equilibrium pressure.

For the evaluation of the production term, we adopt the BGK-type model
for relativistic monatomic and polyatomic gases proposed by Pennisi and
Ruggeri~\cite{PennisiRuggeri2018}, with constant relaxation time $\tau$.
In its general form, the collision term reads
\[
Q=
\frac{u_\alpha p^\alpha}{c^{2}\tau}
\left[
f_E-f
-
f_E\,
\frac{p^\mu q_\mu}{bmc^{2}}
\left(1+\frac{\mathcal{I}}{mc^{2}}\right)
\right],
\]
where $f_E$ denotes the local equilibrium distribution, $q^\mu$ is the
heat flux, and $b$ is a positive coefficient whose explicit expression
is given in Ref.~\cite{PennisiRuggeri2018}. Since the RET$_6$ principal
subsystem satisfies $q^\mu=0$, the collision term reduces to
\[
Q=
\frac{u_\alpha p^\alpha}{c^{2}\tau}(f_E-f).
\]
Accordingly, the trace of the third-order moment and the corresponding
production term are~\cite{Arima2022}
\begin{equation}\label{triplo}
    A^{\alpha\beta}{}_{\beta}
    =
    \bigl[\rho c^2(\theta_{0,2}-\theta_{1,2})+A_1\Pi\bigr]u^\alpha,
    \qquad
    I^\beta{}_{\beta}
    =
    -\frac{A_1}{\tau}\Pi.
\end{equation}
The equilibrium pressure and total energy density are
\begin{equation}\label{statef}
    p:=\frac{\rho c^2}{\gamma},
    \qquad
    e:=\rho c^2\omega,
\end{equation}
where $\omega=\omega(\gamma)$ encodes the thermal equation of state. The
auxiliary functions appearing in the balance laws are
\begin{equation}
    \theta_{0,2}=\omega^2-\omega',
    \qquad
    \theta_{1,2}=\frac{3(\gamma\omega+1)}{\gamma^2},
\end{equation}
where the prime denotes differentiation with respect to $\gamma$ \footnote{In
Ref.~\cite{Arima2022}, p.~20, two typographical errors are present. In
Eq.~(72), the factor $\Pi$ should appear inside the bracket after $A_1$, as in
Eq.~\eqref{triplo} of the present paper. Furthermore, in the definition of
$A_1$ between Eqs.~(72) and (73), the last two terms in the numerator are
written with the opposite sign. These misprints do not affect the results of
Ref.~\cite{Arima2022}.}.
Finally,
\begin{equation}
    A_1 =
    -\frac{
        2\gamma^6\omega^{\prime 3}
        +6\gamma^4\omega^{\prime 2}
        +\left(-\gamma^3\omega''+3\gamma\omega+6\right)^2
        -\gamma^2\omega'
        \left(\gamma^4\omega'''-6\gamma^2\omega^2
        +3\gamma\omega+36\right)}
    {\gamma\left(
        -\gamma^3\omega''
        +\gamma^2(2\gamma\omega-3)\omega'
        +3\gamma\omega+6
    \right)}.
\end{equation}
For the two relativistic caloric equations of state whose classical limits
correspond to monatomic and diatomic gases, respectively, one has
\cite{Arima2022}
\begin{equation}
    \omega_{\rm mono}
    =
    {\cal G}-\frac{1}{\gamma},
    \qquad
    \omega_{\rm dia}
    =
    \frac{1}{\gamma}
    +\frac{\gamma}{\gamma{\cal G}-4},
\end{equation}
with
\begin{equation}
    {\cal G}=\frac{K_3(\gamma)}{K_2(\gamma)},
\end{equation}
where $K_n(\gamma)$ denotes the modified Bessel function of the second
kind of order $n$; see, for instance, \cite{Abramowitz1965handbook}.
Here and in the following, the labels ``monatomic'' and ``diatomic'' refer
to the classical limits of the corresponding relativistic caloric equations
of state: as $\gamma\to\infty$, they recover the specific internal energies
of classical monatomic and diatomic gases, respectively. At finite
relativistic temperatures, however, the specific heats are nonlinear
functions of temperature, owing to the dependence on modified Bessel
functions, and are therefore not constant.

For these two relativistic caloric equations of state, $\omega(\gamma)$ can
be expressed in the closed forms given above, whereas for more general
polyatomic gases its expression involves integrals containing modified Bessel
functions; see Ref.~\cite{Arima2022} for details. In the following, whenever
an explicit equation of state is required, and in particular for the
determination of the hyperbolicity region shown in Fig.~\ref{dominio}, we
specialize to the case $\omega=\omega_{\rm dia}$.

%
%
%
\subsection{Convective form of the system}
\label{sec:convective}
Introducing the material derivative $\dot f := u^\alpha\partial_\alpha f$, which in the context of relativity is simply the total derivative with respect to the proper time of an observer comoving with $u^\alpha$, the system presented in the previous section can be rewritten in a simpler form. 

The continuity equation $\partial_\alpha V^\alpha = 0$ yields
\begin{equation}
    \label{eqs:rho}
	\dot\rho + \rho\,\partial_\alpha u^\alpha = 0 \, .
\end{equation}
The projection of the conservation law of the energy-momentum tensor in the direction of $u^\alpha$, i.e. $u_\beta\,\partial_\alpha T^{\alpha\beta}=0$, reads
\begin{equation}
    \label{eqs:e}
	\dot e + (e+p+\Pi)\,\partial_\alpha u^\alpha = 0 \, ;
\end{equation}
while the projection of the conservation equation for the energy-momentum tensor on the 3-space orthogonal to $u^\alpha$, i.e. $h^{\beta} {}_\gamma\,\partial_\alpha T^{\alpha\gamma}=0$ can be recast as
\begin{equation}
    \label{eqs:a}
    \frac{e+p+\Pi}{c^2}\,\dot u^\beta + h^{\alpha\beta}\partial_\alpha(p+\Pi) = 0 \, .
\end{equation}
Finally, from Eq.~\eqref{eq:trace} and Eq.~\eqref{triplo} one gets the relaxation equation for the dynamical pressure, that reads
\begin{equation}
    \label{eqs:Pi}
	\dot\Pi + \frac{c^2\rho(\theta_{0,2}'-\theta_{1,2}')}{A_1}\,\dot\gamma
    + \frac{\dot A_1}{A_1}\,\Pi + \Pi\,\partial_\alpha u^\alpha = -\frac{\Pi}{\tau} \, .
\end{equation}
Eqs.~\eqref{eqs:rho}--\eqref{eqs:Pi} thus determine the so-called {\em convective form} of the RET$_{6}$ system.

%
%
%
\subsection{Characteristic velocities, hyperbolicity and causality}
\label{sec:charvel}
Besides the material characteristic $\lambda=0$, which has multiplicity four
in the local rest frame, the RET$_6$ system possesses two nonzero longitudinal
characteristic velocities,
\[
    \lambda_{\pm}=\pm\lambda.
\]
Their common squared dimensionless magnitude is given by
(see Appendix~\ref{app:char_vel})
\begin{equation}\label{lambda2a}
    \lambda^2
    =
    \frac{p+\Pi}{e+p+\Pi}
    \left[
        k+\frac{\gamma^2}{\hat c_V}
        \left(\Ltre\bar{\Pi}+\Lquattro\right)
    \right].
\end{equation}
where 
\be 
\label{eq:dimlessPi}
\bar{\Pi} := \frac{\Pi}{\rho \, c^2}
\ee 
is the dimensionless dynamical pressure,
\begin{equation}
    \hat{c}_V = \frac{m c_V}{k_{\rm B}}, \qquad 
    k = \frac{c_p}{c_V} = 1 + \frac{1}{\hat{c}_V}
\end{equation}
being $c_V = d\varepsilon/dT$ and $c_p$, respectively, the specific heat at constant volume and at constant pressure, and  
\begin{equation} \label{eq:L3L4}
    \Ltre := \frac{A_1'}{A_1}, \qquad
    \Lquattro := \frac{\theta_{0,2}^\prime - \theta_{1,2}^\prime}{A_1}\,,
\end{equation}
with $L_1=L_1(\gamma)$ and $L_2=L_2(\gamma)$. 
Taking into account \eqref{eenergia} and \eqref{statef}, we have
\begin{equation}
    \hat{c}_V = -\gamma^2 \omegap, \qquad
    k = 1 - \frac{1}{\gamma^2 \omegap},    
\end{equation}
and Eq.~\eqref{lambda2a} can be rewritten explicitly as
\begin{equation}\label{lambda2}
    \lambda^2 = -\frac{\left(\gamma \bar{\Pi} + 1\right) \left[ -\gamma^2 \omegap + \gamma^2 \left(\Ltre \bar{\Pi} + \Lquattro\right) + 1 \right]}{\gamma^2 \omegap \left( \gamma \, \omega + \gamma \bar{\Pi} + 1 \right)}.
\end{equation}
As is customary in RET, the theory is valid in a neighborhood of the
equilibrium state $\bar{\Pi}=0$. Indeed, the Maximum Entropy Principle gives a
distribution function expressed in terms of Lagrange multipliers, which are
the main field variables symmetrizing the system. Since the map between these
multipliers and the physical variables is nonlinear, the explicit closure is
obtained by expanding the distribution function near equilibrium and
linearizing this map with respect to the non-equilibrium variables. The
resulting closure is therefore local: the approximated distribution function
need not remain positive far from equilibrium, and the convexity of the
entropy, hence symmetric hyperbolicity, is guaranteed only in a suitable
neighborhood of equilibrium; see \cite{ET,MuellerRuggeri1998,Arima2022}.
At the same time, not all restrictions of the model should be interpreted as
mere consequences of hyperbolicity. Some of them are inherited directly from
the underlying kinetic construction and from the admissibility properties of
the distribution function. This point will be important below, where the
no-go result follows from the kinetic structure of the model rather than from
a loss of hyperbolicity alone.

At equilibrium, the RET$_6$ longitudinal characteristic velocity satisfies
$0<\lambda_{\mathrm{Eq}}^2<1$, so that the equilibrium system is both
hyperbolic and causal.
For the diatomic equation of state $\omega=\omega_{\rm dia}$, the
physical hyperbolicity and causality domain
\[
\{\,0<\lambda^2<1\,\}
\]
is the region depicted in Fig.~\ref{dominio}, delimited by the two continuous curves.
 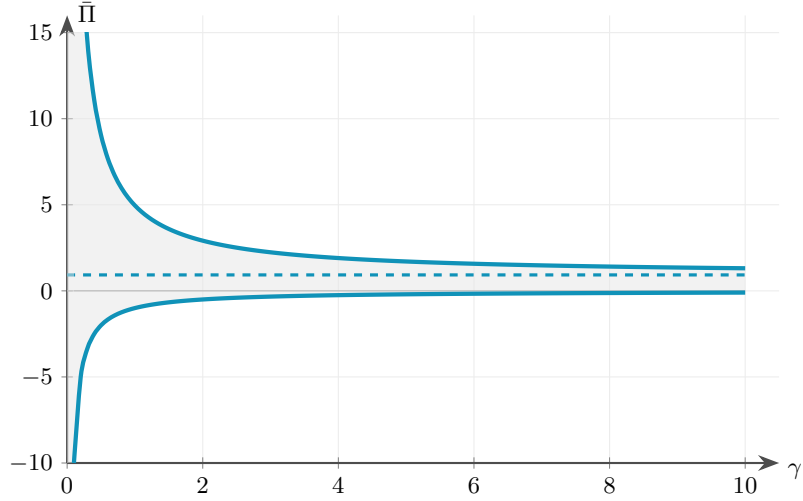
\begin{figure}
	\centering
	\begin{tikzpicture}
		\begin{axis}[
			axis lines = left,
			axis line style = {-{Stealth[scale=1.2]}, thick, gray!60!black},
			xlabel = {$\gamma$},
			ylabel = {$\bar\Pi$},			
			xlabel style = {
				at={(axis description cs:1,0)},
				anchor=north west, 
				font=\small,
				xshift=0cm,
				yshift=0.1cm
			},
			ylabel style = {
				at={(axis description cs:0,1)},
				anchor=south east, 
				font=\small, 
				rotate=-90,
				xshift=0.5cm,
				yshift=-0.2cm
			},
			xmin = 0, xmax = 10.5,
			ymin = -10, ymax = 16,
			grid = major,
			grid style = {thin, gray!15},
			legend cell align = {left},
			legend style = {
				at={(0.95,0.95)}, 
				anchor=north east, 
				draw=gray!40, 
				fill=white, 
				font=\footnotesize,
				row sep=2pt
			},
			width = 11cm,
			height = 7.5cm,
			restrict y to domain=-12:45,
			clip = true
			]
			
			\addplot[gray!60, forget plot] coordinates {(0,0) (10,0)};
			
			\addplot[
			dashed, 
			cyan!70!black,
			very thick, 
			domain=0:10, 
			forget plot
			] {25/27};
			
			\addplot[
			name path=ramoPosSoffitto,
			forget plot,
			opacity=0,
			smooth
			] table[y expr={min(\thisrowno{1}, 15)}] {fig_boundary_pos.dat};
			
			\addplot[
			name path=ramoNeg,
			forget plot,
			opacity=0,
			smooth
			] table {fig_boundary_neg.dat};
			
			\addplot[
			fill=gray!20,     
			opacity=0.5,      
			forget plot        
			] fill between [
			of=ramoPosSoffitto and ramoNeg, 
			on layer={pre main} 
			];
			
			\fill[gray!20, opacity=0.5] 
			(axis cs:0, -10) -- (axis cs:0.1, -10) -- (axis cs:0.1, 15) -- (axis cs:0, 15) -- cycle;
			
			\clip (axis cs:0,-10) rectangle (axis cs:10.5,15);
			
			\addplot[
			color=cyan!70!black,
			ultra thick,
			smooth,
			mark=none
			] table {fig_boundary_pos.dat};
			
			\addplot[
			color=cyan!70!black,
			ultra thick,
			smooth,
			mark=none
			] table {fig_boundary_neg.dat};			
		\end{axis}
	\end{tikzpicture}
    \caption{Region in the $(\gamma,\bar{\Pi})$-plane for which $0 < \lambda^2 < 1$.
		 The lower curve marks the boundary $p+\Pi=0$; the dashed line
		$\bar{\Pi} = 25/27$ is the upper bound and $\bar{\Pi}=0$ the lower bound
		in the limit $\gamma\to\infty$.}                            
	\label{dominio}
\end{figure}
In the ultrarelativistic limit, i.e. for $\gamma\to 0$, Eq.~\eqref{lambda2} reduces to
\begin{equation}
	\lambda^2 = \frac{1}{3} + \frac{\gamma\bar{\Pi}}{6},
\end{equation}
while in the classical limit, i.e. for $\gamma\to\infty$, restoring dimensional velocity $\tilde\lambda = c\lambda$, one recovers the classical RET$_6$ result \cite{ET6}
\begin{equation}
	\tilde\lambda^2 = \frac{5}{3}\,\frac{p+\Pi}{\rho}.
\end{equation}
%
%
%
\subsection{Subcharacteristic condition \label{sec:subchar}}
The Euler system is a principal subsystem of the RET$_6$ model. Since the entropy is convex at equilibrium, by a theorem stated in \cite{Arma} the subcharacteristic condition holds, namely
\begin{equation}
    \lambda_{\mathrm{Eq}} > \lambda_{\mathrm{Eul}}, \qquad \forall \gamma>0,
\end{equation}
where $\lambda_{\mathrm{Eq}}$ denotes the equilibrium value of \eqref{lambda2}, and $\lambda_{\mathrm{Eul}}$ is the characteristic velocity of the Euler fluid \cite{JMP}. Explicitly, one has
\begin{equation}
    \lambda_{\mathrm{Eq}}^{2} = \frac{p}{e+p} \left( k + \Lquattro\,\frac{\gamma^{2}}{\hat{c}_V} \right),
    \qquad
    \lambda_{\mathrm{Eul}}^{2} = k\,\frac{p}{e+p}.
\end{equation}
Indeed, the inequality follows from the positivity of the additional term, since
\begin{equation}
    \Lquattro\,\frac{\gamma^{2}}{\hat{c}_V} > 0, \qquad \forall \gamma > 0.    
\end{equation}
%
%
%
\section{RET$_6$ on curved spacetimes}
\label{sec:FLRW}
According to the {\em minimal coupling principle}, the standard way to pass
from the special-relativistic equations to the generally covariant ones
consists of: (i) replacing the Minkowski metric with the spacetime metric,
i.e. $\eta_{\mu\nu}\to g_{\mu\nu}$; (ii) replacing partial derivatives with
covariant derivatives, i.e. $\partial_\mu\to\nabla_\mu$; (iii) requiring that
the local structure of the constitutive equations remains the same as in
special relativity, so that curvature tensors do not enter explicitly in the
constitutive relations of the generally covariant uplift of the theory.

A related kinetic treatment of relativistic polyatomic gases in a gravitational
field was developed in \cite{OliveiraRamosSoares2022}, where Minkowski and
FLRW spacetimes were considered and the dynamical pressure was obtained by a
Chapman--Enskog expansion. That work is based on an earlier relativistic
polyatomic kinetic model \cite{PennisiRuggeri2017}. Here, consistently with
Sec.~\ref{sec:RET}, we use the RET$_6$ model derived from the more recent
relativistic moment hierarchy \cite{Arima2022}, and we retain its hyperbolic
balance-law structure rather than passing to a parabolic Eckart-type
Chapman--Enskog regime.

The generally covariant convective form of RET$_6$ then reads
\begin{align}
    &\dot\rho + \Theta\rho = 0, \label{eq:cont_FLRW} \\
    &\dot e + \left(e + p + \Pi\right) \Theta = 0,
        \label{eq:energy_FLRW} \\
    &\frac{e + p + \Pi}{c^2}\,\dot u^\beta
        + h^{\alpha\beta}\nabla_\alpha\left(p + \Pi\right) = 0,
        \label{eq:euler_FLRW} \\
    &\dot\Pi
        + c^2\rho \Lquattro\,\dot\gamma
        + \Ltre \dot\gamma \,\Pi
        + \Pi\,\Theta
        = -\frac{\Pi}{\tau}.
        \label{eq:Pi_FLRW}
\end{align}
where $\dot\psi = u^\alpha\nabla_\alpha\psi$ for any scalar field $\psi$,
$\Theta := \nabla_\alpha u^\alpha$ is the expansion scalar, and
$h^{\alpha\beta} = u^\alpha u^\beta/c^2 - g^{\alpha\beta}$.

The FLRW metric in comoving coordinates reads
\begin{equation}
	\mathrm{d}s^2 = c^2\mathrm{d}t^2
	- a^2\left(t\right) \left(\frac{\mathrm{d}r^2}{1-Kr^2} + r^2\mathrm{d}\Omega^2\right),
\end{equation}
with $a(t)$ being the scale factor, $H := \dot a/a$ is the Hubble parameter,
$K$ is the constant spatial curvature parameter, and $\Theta=3H$. Using
$x^0=ct$ as the temporal coordinate, the four-velocity of observers comoving
with the cosmological fluid is $u^\mu=(c,\bm0)$. Consequently,
$u^\alpha\nabla_\alpha\psi={\rm d}\psi/{\rm d}t$ for every homogeneous
scalar field $\psi$, so that the dot reduces to differentiation with respect
to the comoving time $t$. Moreover, $\dot u^\alpha=0$ and
$h^{\alpha\beta}\nabla_\alpha(p+\Pi)=0$. Hence, describing the cosmic fluid
using RET$_6$ allows us to rewrite the system
\eqref{eq:cont_FLRW}--\eqref{eq:Pi_FLRW} as 
\begin{align}
	&\dot\rho + 3H\rho = 0, \label{eq:cont_red} \\
	&\dot e + 3H\left(e + p + \Pi\right) = 0, \label{eq:energy_red} \\
	&\dot\Pi + \left(\Ltre \dot\gamma + 3H + 
        \frac{1}{\tau}\right) \Pi + c^2\rho \Lquattro\,\dot\gamma = 0 \, . \label{eq:Pi_red}
\end{align}
Thus, the RET$_6$ acts as a source term of the Einstein field equations that, specialized to FLRW cosmology read \cite{Ellisetal:2012}
\begin{align}
	& \frac{\ddot a}{a} = -\frac{4\pi G}{3c^2}\big[e + 3\left(p + \Pi\right)\big] \, , \label{eq:Friedmann1b} \\
    & H^2 + \frac{K c^2}{a^2} = \frac{8\pi G}{3c^2} e \, ,
	\label{eq:Friedmann2b}
\end{align}
known respectively as the Raychaudhuri equation (also known as the second Friedmann equation, or as the acceleration equation) and the Friedmann constraint equation (or energy constraint, or also as the first Friedmann equation), with $G$ denoting Newton's gravitational constant.

Current cosmological observations strongly constrain the spatial curvature to
be close to zero (see e.g. \cite{Ellisetal:2012,Planck2020}), and in the
following we restrict our discussion to spatially flat FLRW geometries, i.e. we
set $K=0$.

It is now convenient to rewrite the full system in dimensionless form. We
use $t=0$ to denote an arbitrary reference epoch and set
\begin{equation}
    a_0:=a(0), \qquad H_0:=H(0), \qquad \bar t:=H_0t.
\end{equation}
Thus, the subscript $0$ denotes quantities evaluated at the chosen reference
(or initial) epoch; present-day quantities used below will instead be denoted
by an asterisk. We introduce
\begin{equation}
    e_0 := \rho_0 c^2\,\omega_0, \qquad
    \omega_0 := \omega\left(\gamma_0\right), \qquad
    \gamma_0 := \frac{mc^2}{k_{\rm B} \,T_0}, \qquad
    T_0 := T(0).
\end{equation}
The dimensionless matter density parameter at the reference epoch is
\begin{equation}
\label{eq:OmegaM0}
    \Omega_{{\rm M},0} := \frac{8\pi G}{3 c^2 H_0^2} \,e_0.
\end{equation}
Integration of the continuity equation $\dot\rho+3H\rho=0$ gives
\begin{equation}
    \rho(t)=\rho_0\left(\frac{a_0}{a(t)}\right)^3,
    \label{eq:rho_a}
\end{equation}
and therefore
$e=\rho_0c^2(a_0/a)^3\omega(\gamma)$. Combining this relation with
Eqs.~\eqref{eq:Friedmann2b} and \eqref{eq:OmegaM0}, we obtain
\begin{equation}
    \mathcal{H}(a,\gamma) := \frac{H}{H_0}
    = \sqrt{\Omega_{{\rm M},0}
    \left(\frac{a_0}{a}\right)^3
    \frac{\omega(\gamma)}{\omega_0}},
    \label{eq:Hhat0}
\end{equation}
which satisfies $\mathcal{H}(a_0,\gamma_0)=1$. For the spatially flat
matter-only model considered in this section, the Friedmann constraint at the
reference epoch gives $\Omega_{{\rm M},0}=1$.

Now, from Eqs.~\eqref{eq:dimlessPi} and \eqref{statef} we can easily infer that
\be 
p + \Pi = \rho c^2\!\left(\frac{1}{\gamma} + \bar{\Pi}\right).
\ee
Thus the rate of change of $\gamma$ is obtained from the energy
equation~\eqref{eq:energy_red}, by substituting 
$e = \rho_0 c^2(a_0/a)^3\omega(\gamma)$ and recalling that
$\dot\rho = -3H\rho$. This yields
\begin{equation}
    \frac{\dot\gamma}{H_0} =
    -\frac{3\,\mathcal{H}(a,\gamma)}{\omega'(\gamma)}
    \left(\frac{1}{\gamma} + \bar{\Pi}\right).
    \label{eq:gdot_Lambda}
\end{equation}
Furthermore, the relaxation equation~\eqref{eq:Pi_red}, expressed in terms of $\bar{\Pi}$ {\em via} $\dot\Pi = \rho c^2\dot{\bar{\Pi}} + \dot\rho c^2\bar{\Pi}$ and $\dot\rho = -3H\rho$, becomes
\begin{equation}
    \frac{{\rm d}\bar{\Pi}}{{\rm d} \bar{t}} =
    -\left(\frac{\Ltre\,\dot\gamma}{H_0} + \frac{1}{\bar\tau}\right)\,\bar{\Pi}
    - \frac{\Lquattro\,\dot\gamma}{H_0},
    \label{eq:ode_pi_Lambda}
\end{equation}
where $\bar\tau := \tau H_0$ is the dimensionless relaxation time. 

Substituting~\eqref{eq:gdot_Lambda} into~\eqref{eq:ode_pi_Lambda}, the full
system \eqref{eq:cont_red}--\eqref{eq:Pi_red} can be rewritten in dimensionless form as
\begin{align}
    \frac{{\rm d} a}{{\rm d}\bar{t}} &= \mathcal{H}(a,\gamma)\,a,
    \label{eq:ode_a_Lambda}\\[4pt]
    \frac{{\rm d}\gamma}{{\rm d}\bar{t}} &=
    -\frac{3\,\mathcal{H}(a,\gamma)}{\omega'(\gamma)}
    \left(\frac{1}{\gamma}+\bar{\Pi}\right),
    \label{eq:ode_gamma_Lambda}\\[4pt]
    \frac{{\rm d}\bar{\Pi}}{{\rm d}\bar{t}} &=
    \frac{3\,\mathcal{H}(a,\gamma)}{\omega'(\gamma)}
    \left(\frac{1}{\gamma}+\bar{\Pi}\right)
    \!\bigl(\Ltre\,\bar{\Pi} + \Lquattro\bigr)
    - \frac{\bar{\Pi}}{\bar\tau}.
    \label{eq:ode_pi_Lambda2}
\end{align}
Note that the case of the Euler fluid is recovered from the latter system by setting $\bar{\Pi} = 0$ and eliminating the last equation. The resulting Euler system reads
\begin{align}\label{eq:euler}
\begin{split}
   &   \frac{da}{d\bar{t}} = \mathcal{H}(a,\gamma)\,a, \\
   & \frac{d\gamma}{d\bar{t}}  =
    -\frac{3\,\mathcal{H}(a,\gamma)}{\gamma \, \omega'(\gamma)}.
\end{split}   
\end{align}

Notably, it is often worth rewriting our system of equations using the scale factor as independent variable. This can be done by observing that
$$
\frac{{\rm d}}{{\rm d} \bar{t}} = a \, \mathcal{H} \, \frac{{\rm d}}{{\rm d} a} \, .
$$
For instance, this procedure yields the system
\begin{align}
    \frac{{\rm d}\gamma}{{\rm d}a} &=
    -\frac{3}{a\,\omega'(\gamma)}
    \left(\frac{1}{\gamma} + \bar{\Pi}\right),
    \label{eq:ode_gamma_a}\\[4pt]
    \frac{{\rm d} \bar{\Pi}}{{\rm d} a} &=
    \frac{3}{a\,\omega'(\gamma)}
    \left(\frac{1}{\gamma}+\bar{\Pi}\right)
    \bigl(L_1\,\bar{\Pi} + L_2\bigr)
    - \frac{\bar{\Pi}}{\bar\tau\,\mathcal{H}(a,\gamma)\,a},
    \label{eq:ode_pi_a}
\end{align}
that, in the case of the Euler fluid reduces to a single equation, i.e.
\begin{equation}
    \frac{{\rm d}\gamma}{{\rm d}a} =
    -\frac{3}{a\,\gamma \,\omega'(\gamma)}.
\end{equation}
This last equation for the Euler fluid can be easily integrated and yields
\be 
a(\gamma) = a_0 \, \exp \left[ \frac{1}{3} \left( - \gamma \, \omega(\gamma) + \gamma_0 \, \omega_0 + \int_{\gamma_0} ^\gamma \omega (\xi) \, {\rm d}\xi \right) \right].
\ee
This expression differs from a power-law behavior which is typical of a FLRW universe that is filled with a single, adiabatically expanding perfect fluid obeying a constant barotropic equation of state, representing the typical toy scenario in cosmology often used to study exact (closed form) solutions.

%
%
%
%
\section{A polyatomic kinetic-theory no-go theorem}
\label{sec:nogo}
In this section we establish a structural constraint imposed by the kinetic
origin of the macroscopic stress-energy tensor. For ordinary monatomic
relativistic kinetic matter, the fact that a stress-energy tensor induced by a
non-negative distribution function satisfies the standard energy conditions is a
classical result of relativistic kinetic theory and of the Einstein--Vlasov
framework
\cite{Ehlers1971,HawkingEllis1973,Wald1984,Andreasson2011,SarbachZannias2013}.
The result needed here is, however, not simply the monatomic statement. In the
present theory the distribution function is defined on the extended polyatomic
phase space,
\[
    f=f(x^\alpha,p^\beta,\mathcal I),
\]
where the additional variable $\mathcal I$ accounts for the internal molecular
degrees of freedom, and the stress-energy tensor contains the polyatomic weight
\[
    \left(mc^2+\mathcal I\right)\phi(\mathcal I).
\]
Thus, although the algebraic mechanism of the proof is analogous to the
classical one, the tensor to which it is applied is the polyatomic kinetic
stress-energy tensor underlying the RET hierarchy.

The theorem below shows that the same local positivity mechanism survives in
the polyatomic setting. The argument is independent of the particular finite
moment closure and therefore applies to RET$_{15}$, to its principal subsystems,
and to higher-order moment theories, as long as the stress-energy tensor is
induced by a non-negative polyatomic distribution function. The specific role of
RET$_6$ appears only in the FLRW application. By homogeneity and isotropy, heat
flux and shear stress vanish, and the only remaining dissipative scalar is the
dynamical pressure $\Pi$.
\begin{theorem}[Polyatomic kinetic-theory no-go theorem]
\label{th-1}
Let $x$ be an arbitrary spacetime event and consider a local inertial frame at
$x$. Let
\[
    f=f(x^\alpha,p^\beta,\mathcal I)\geq 0
\]
be an admissible one-particle distribution function for a relativistic
polyatomic gas. Assume that $p^\alpha$ is future-directed and satisfies the
mass-shell condition
\[
    p^\alpha p_\alpha=m^2c^2
\]
in the local frame, and let $\phi(\mathcal I)\geq0$ be the density of internal
states. If the stress-energy tensor is given by the second kinetic moment in
\eqref{eq:RET15_moments_def}, then $T^{\alpha\beta}$ satisfies the strong energy
condition pointwise at the event $x$.

Moreover, when the corresponding relativistic polyatomic gas model is extended
to homogeneous and isotropic curved spacetimes through the minimal-coupling
prescription, the same local argument applies in every local inertial frame. In
particular, when such a relativistic polyatomic gas is used as the matter source
in an FLRW cosmological model, the stress-energy tensor takes the form
\begin{equation}
\label{eq:Tiso_nogo}
    T^{\alpha\beta}
    =
    \frac{e}{c^2}u^\alpha u^\beta
    +
    (p+\Pi)h^{\alpha\beta},
\end{equation}
and, for non-vacuum matter states,
\begin{equation}
\label{eq:SEC_FLRW}
    e>0,
    \qquad
    p+\Pi\geq0,
    \qquad
    e+3(p+\Pi)>0 .
\end{equation}
This conclusion is independent of the particular closure and hence holds for
RET$_{15}$, for its principal subsystems, and for higher-order moment theories
with the same polyatomic kinetic stress-energy tensor.
\end{theorem}
\begin{proof}
Fix a spacetime event $x$ and choose a local inertial frame at that point. Let
$v^\beta$ be any future-directed timelike vector. Using the second kinetic
moment in \eqref{eq:RET15_moments_def}, the covariant strong-energy-condition
combination can be written as
\begin{align}
&\left(
T_{\beta\theta}
-\frac{1}{2}T^\nu{}_\nu g_{\beta\theta}
\right)v^\beta v^\theta
\nonumber\\
&\qquad =
\frac{1}{mc}
\int_{\mathbb R^3}\!\!\int_0^{+\infty}
    f
    \left[
        (p_\beta v^\beta)^2
        -\frac{1}{2}(p_\nu p^\nu)(v_\theta v^\theta)
    \right]
    \left(mc^2+\mathcal I\right)
    \phi(\mathcal I)\,d\mathcal I\,d\boldsymbol P .
\label{eq:SEC_covariant_nogo_proof}
\end{align}
With the signature $(+,-,-,-)$, both $p^\alpha$ and $v^\alpha$ are
future-directed timelike vectors. Hence the reverse Cauchy--Schwarz inequality
gives
\[
    (p_\alpha v^\alpha)^2
    \geq
    (p_\alpha p^\alpha)(v_\beta v^\beta),
\]
and therefore
\[
    (p_\alpha v^\alpha)^2
    -\frac{1}{2}(p_\alpha p^\alpha)(v_\beta v^\beta)
    \geq
    \frac{1}{2}(p_\alpha p^\alpha)(v_\beta v^\beta)>0.
\]
Since
\[
    f\geq0,
    \qquad
    \phi(\mathcal I)\geq0,
    \qquad
    mc^2+\mathcal I>0,
\]
the integrand in \eqref{eq:SEC_covariant_nogo_proof} is non-negative.
Consequently, the strong energy condition holds pointwise. Moreover, for a
non-vacuum state, namely when
$f\,\phi(\mathcal I)$ is nonzero on a set of positive measure in phase space,
the integral is strictly positive for every future-directed timelike vector
$v^\alpha$.

We now specialize the same conclusion to the homogeneous and isotropic form used
in the FLRW application. From the decomposition
\[
    T^{\alpha\beta}
    =
    \frac{e}{c^2}u^\alpha u^\beta
    +
    (p+\Pi)h^{\alpha\beta}
\]
and from the second kinetic moment in \eqref{eq:RET15_moments_def}, the physical
energy density and the effective isotropic pressure are obtained by projection:
\begin{align}
    e
    &=
    \frac{1}{c^2}T^{\alpha\beta}u_\alpha u_\beta
    \nonumber\\
    &=
    \frac{1}{mc^3}
    \int_{\mathbb R^3}\!\!\int_0^{+\infty}
    f\,(p^\alpha u_\alpha)^2
    \left(mc^2+\mathcal I\right)
    \phi(\mathcal I)\,d\mathcal I\,d\boldsymbol P
    >0,
    \label{eq:e_positive_nogo}
    \\[6pt]
    p+\Pi
    &=
    \frac{1}{3}T^{\alpha\beta}h_{\alpha\beta}
    \nonumber\\
    &=
    \frac{1}{3mc}
    \int_{\mathbb R^3}\!\!\int_0^{+\infty}
    f\,p^\alpha p^\beta h_{\alpha\beta}
    \left(mc^2+\mathcal I\right)
    \phi(\mathcal I)\,d\mathcal I\,d\boldsymbol P
    \geq0 .
    \label{eq:pressure_positive_nogo}
\end{align}
Here $h_{\alpha\beta}$ is the positive spatial projector in the local rest space
orthogonal to $u^\alpha$. Therefore
\[
    e+p+\Pi>0,
    \qquad
    e+3(p+\Pi)>0 .
\]
Finally, the minimal-coupling principle preserves the local constitutive
structure of the fluid. Since the energy conditions are purely scalar algebraic  inequalities (constructed from tensor contractions), their validity in a local inertial frame extends to any coordinate system at the same spacetime point $x$. Consequently, the positivity conditions derived in Minkowski spacetime hold pointwise on the tangent space at every event of the curved spacetime. Because the strong energy condition is a local algebraic condition on the stress-energy tensor, the same conclusion applies to the general-covariant extension of the model obtained through minimal coupling.
\end{proof}

We can now present an immediate consequence of Theorem~\ref{th-1} concerning
the effect of polyatomic kinetic matter on cosmic dynamics. To this end we
introduce the deceleration parameter as \cite{Ellisetal:2012}
\begin{equation}
\label{eq:decel}
    q :=
    -\frac{\ddot a\,a}{\dot a^2}
    =
    -\frac{\ddot a}{aH^2}.
\end{equation}
\begin{corollary}[No accelerated expansion]
Consider a spatially flat FLRW universe sourced by a relativistic polyatomic gas
as described in Theorem~\ref{th-1}. In particular, this applies to RET$_6$ in
its admissible kinetic regime. Then the scale factor $a(t)$ necessarily
satisfies
\[
    \ddot a\leq0,
\]
while the deceleration parameter satisfies
\[
    q\geq\frac{1}{2}>0.
\]
In other words, accelerated expansion cannot arise from such polyatomic kinetic
matter alone.
\end{corollary}
\begin{proof}
The first part of the statement follows directly from the Raychaudhuri equation
\eqref{eq:Friedmann1b} and Theorem~\ref{th-1}. The second part follows by
combining Eq.~\eqref{eq:Friedmann1b} and Eq.~\eqref{eq:Friedmann2b} with
$K=0$, which yields
\begin{equation}
\label{eq:q_noLambda}
    q
    =
    \frac{1}{2}\left[
    1+\frac{3\bigl(p+\Pi\bigr)}{e}
    \right]
    =
    \frac{1}{2}\left[
    1+\frac{3\bigl(1/\gamma+\bar{\Pi}\bigr)}{\omega(\gamma)}
    \right].
\end{equation}
The conclusion follows from $e>0$ and $p+\Pi\geq0$, as established in
Theorem~\ref{th-1}.
\end{proof}

Hence, the non-equilibrium pressure may reduce the deceleration
parameter, but it cannot make the expansion accelerated as long as the matter
source remains kinetically realizable by a non-negative polyatomic distribution
function. Therefore, an accelerated expansion of the universe still requires physics beyond standard relativistic kinetic matter of massive particles with internal degrees of freedom, as one would naively expect. The present result makes this obstruction explicit for the polyatomic RET closure used in this paper.
\section{Including the cosmological constant: the $\Lambda$RET$_{6}$ model}
\label{sec:LRET}
The analysis carried out so far deliberately omitted the cosmological constant $\Lambda$ in order to isolate the purely kinetic-theory effects of the bulk viscous pressure of the RET$_6$ fluid. We now reintroduce $\Lambda$ in the Friedmann equations and comment on the effect of the RET$_6$ fluid combined with the dark energy fluid.

First, observe that for a flat ($K=0$) FLRW universe, reinstating the cosmological constant in the Friedmann equations yields
\begin{align}
    H^2 &= \frac{8\pi G}{3c^2}\,e + \frac{\Lambda c^2}{3},
    \label{eq:Friedmann1_Lambda} \\
    \frac{\ddot a}{a} &= -\frac{4\pi G}{3c^2}\bigl[e + 3(p+\Pi)\bigr]
    + \frac{\Lambda c^2}{3}.
    \label{eq:Friedmann2_Lambda}
\end{align}
If we now consider a positive cosmological constant, an accelerated expansion of the universe is achieved ($\ddot{a} > 0$) if the condition
\begin{equation}
    \frac{\Lambda c^2}{3} > \frac{4\pi G}{3c^2}\bigl[e + 3(p+\Pi)\bigr] \, ,
\end{equation}
holds. Hence, as already pointed out in Section \ref{sec:nogo}, the cosmological constant (or an alternative form of dark energy) remains \emph{necessary} for acceleration even in the RET$_6$ framework. The dynamical pressure $\Pi$ alone cannot replace $\Lambda$, but it does modify the expansion history.

Consider the dimensionless density parameters
\begin{equation}
    \Omega_{{\rm M},0} := \frac{8\pi G\,e_0}{3 c^2 H_0^2}, \qquad
    \Omega_{\Lambda,0} := \frac{\Lambda c^2}{3 H_0^2} \, ,
\end{equation}
with the flat-universe constraint
\begin{equation}\label{condz}
    \Omega_{{\rm M},0} + \Omega_{\Lambda,0} = 1 \, ,
\end{equation}
it is easy to verify that the evolution equations are still those of the system  
\eqref{eq:ode_a_Lambda}--\eqref{eq:ode_pi_Lambda2} with the dimensionless Hubble parameter modified as
\begin{equation}
    \mathcal{H}(a,\gamma) := \frac{H}{H_0}
    = \sqrt{\Omega_{{\rm M},0}
      \left(\frac{a_0}{a}\right)^3
      \frac{\omega(\gamma)}{\omega_0}
      + \Omega_{\Lambda,0}}.
    \label{eq:Hhat}
\end{equation}

We can also easily compute the modifications to the expansion parameter seen in Section \ref{sec:nogo}. Indeed, substituting \eqref{eq:Friedmann2_Lambda} and eliminating $\Lambda$ from \eqref{eq:Friedmann1_Lambda} we find
\begin{align}
\label{eq:decparlambda}
q=-\frac{\ddot{a}}{a H^2} &= -\frac{1}{  H^2}\left\{-\frac{4\pi G}{3 c^2}[e+3(p+\Pi)]+H^2-\frac{8\pi G}{3 c^2}e\right\}\\
&= \frac{4\pi G}{ H^2 c^2}(  e+ p+\Pi ) - 1\\
&= \frac{4\pi G \rho}{  H^2  } \left[\omega(\gamma)+\frac{1}{\gamma}+\bar{\Pi} \right] - 1 \\
& = \frac{3}{2\mathcal{H}^2}
\left[\Omega_{{\rm M},0}
\left(\frac{a_0}{a}\right)^3
\frac{\omega(\gamma)}{\omega_0}\right]
\left(1+\frac{1/\gamma+\bar\Pi}{\omega(\gamma)}\right)-1
\end{align}

Comparing this result [Eq.~\eqref{eq:decparlambda}] with the deceleration parameter of the $\Lambda$-plus-dust model, i.e.
$$
q_{\Lambda + {\rm dust}} = \frac{4\pi G \rho_{\rm m}}{H^2} - 1,
$$
assuming $\rho = \rho_{\rm m}$ (with $\rho_{\rm m}$, the mass density of the dust), at a fixed value of $H$ yields
\be 
\Delta q = (q - q_{\Lambda + {\rm dust}})\big|_{H \, {\rm fixed}} = \frac{4\pi G \rho}{H^2} \left[\omega(\gamma)+\frac{1}{\gamma}+\bar{\Pi} - 1 \right] \, .
\ee
Note that, at late time, the $\Lambda$ plus dust model approximates the $\Lambda$CDM model since the contribution of radiation is sub-dominant in this regime. A brief discussion of the extension of $\Lambda$RET$_6$ that includes also the contribution of radiation is deferred to the end of Sec.~\ref{sec:numerical}. However, a complete phenomenological characterization of this model is beyond the scope of this work and is left for future investigations.
\subsection{Stability of the $\Lambda$RET$_6$ Model}
\label{StabLRET6}
We now specialize to the diatomic equation of state and assume a constant
positive relaxation time $\tau>0$. We employ standard dynamical-system methods
to investigate the existence and stability of a de~Sitter attractor for this
model; see \cite{Coley:2003mj} and references therein.

Consider the dynamical system of equations \eqref{eq:ode_a_Lambda}--\eqref{eq:ode_pi_Lambda2} with the dimensionless Hubble parameter as in \eqref{eq:Hhat}. To properly analyze the late-time behavior and avoid a diverging scale factor and temperature, we introduce the inverse variables $x = 1/a$ and $\xi = 1/\gamma$. Then, with respect to the physical cosmic time $t$ and introducing the notation $\omega_{,\xi} = \partial \omega / \partial \xi$, the system takes the following form:
\begin{align}
\dot{x} &= - {H} (x,\xi) \, x \ , \\
\dot{\xi} &= -\frac{3 {H}(x,\xi)}{\omega_{,\xi}} \left(\bar{\Pi} +\xi\right) \ , \\
\dot{\bar{\Pi}} &=-\frac{3 {H} (x,\xi)}{\xi^2\,\omega_{,\xi}} \left(\bar{\Pi} +\xi\right)[L_1 (\xi) \bar{\Pi}+L_2 (\xi)]-\frac{\bar{\Pi}}{{\tau}}.
\label{dsAsystem}
\end{align}
Setting the right-hand sides of the system to zero, the relevant fixed point corresponding to the de Sitter attractor is located at the origin of the new phase space, namely $(x=0, \, \xi=0, \, \bar{\Pi}=0)$ as shown in Figure~\ref{fig:dSAttractor}.
\begin{figure}
    \centering
    \includegraphics[width=0.55\linewidth]{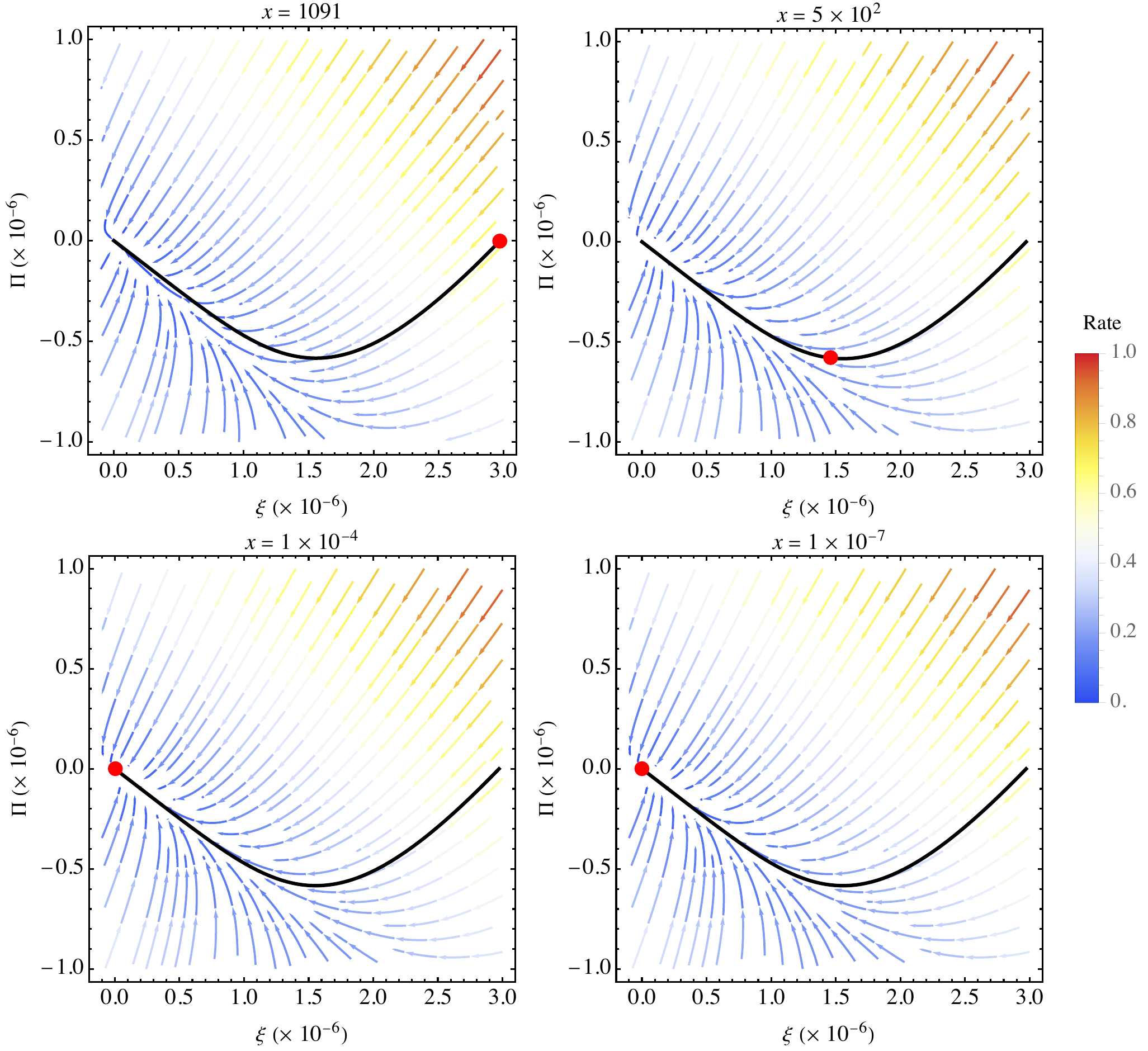}
    \caption{Phase portrait of the de Sitter attractor for the system \eqref{dsAsystem}, projected in the $\xi(x)$-$\bar \Pi(x)$ plane, evaluated at decreasing values of $x=1/a$ (i.e., increasing cosmic time). The system (red dot) dynamically evolves toward the origin $(0,0,0)$, starting from the recombination era ($a_0=1/1091, \xi_0 \approx 2.7\times 10^{-10}, \bar {\Pi}_0\approx -2.5 \times 10^{-10}$). The color gradient indicates the relative rate of the evolution in the phase space, taking normalized values in the interval $[1,0]$. Axes have been rescaled for the sake of clarity.}
    \label{fig:dSAttractor}
\end{figure}

Notice that, at this fixed point, the matter density vanishes (see Eq.~\eqref{eq:rho_a}) and the Hubble parameter reduces to that of de Sitter space, i.e.
\be
H(x,\xi)\big|_{x=0, \, \xi=0} = c\sqrt{\frac{\Lambda}{3}} = H_{\text{dS}}.
\ee

The stability of the system at late times is determined by the signs of the eigenvalues of the Jacobian matrix evaluated at this fixed point. The Jacobian $J$ at the de Sitter fixed point reads (see Appendix~\ref{app:jacobian} for details):
\be
J=
    \begin{pmatrix}
 -H_{\text{dS}} & 0 & 0 \\[1.2em]
 0 & -\dfrac{6H_{\text{dS}}}{5} & -\dfrac{6 H_{\text{dS}}}{5} \\[1.2em]
 0 & -\dfrac{4 H_{\text{dS}}}{5} & -\dfrac{4 H_{\text{dS}} \tau +5}{5 \tau } 
    \end{pmatrix} \ .
\ee
The corresponding eigenvalue problem for the eigenvalues $\mu$ yields:
\be
\mu_1 = -H_{\text{dS}} = -c\sqrt{\dfrac{\Lambda}{3}}, \quad \quad \mu_{2,3}=\frac{1}{2\tau}\left(-2 H_{\text{dS}} \tau -1 \pm \frac{\sqrt{20 H_{\text{dS}}^2 \tau ^2-4 H_{\text{dS}} \tau +5}}{\sqrt{5}} \right) \ .
\ee
Note that the argument of the square root is strictly positive for any physical value of $H_{\text{dS}}>0$ and $\tau >0$, guaranteeing that $\mu_{2,3}$ are well-defined and real. The decoupling of the spatial sector trivially yields $\mu_1 < 0$. Furthermore, it is manifest that $\mu_3$ (associated with the minus sign) is strictly negative. 

To ensure the complete asymptotic stability of the thermodynamic sector we must verify that $\mu_2 < 0$, which requires
\be
\frac{1}{2\tau}\left(-2 H_{\text{dS}} \tau -1 + \frac{\sqrt{20 H_{\text{dS}}^2 \tau ^2-4 H_{\text{dS}} \tau +5}}{\sqrt{5}} \right)<0 \, .
\ee
A straightforward calculation then yields
\be
24 H_{\text{dS}} \tau > 0.
\ee
Since $H_{\text{dS}}>0$ and $\tau>0$, this condition is identically satisfied. This formally guarantees the (local) stability of the de Sitter attractor at late times.

\section{Numerical Simulations}
\label{sec:numerical}

All numerical integrations presented in this section use the diatomic
equation of state $\omega=\omega_{\rm dia}$. The monatomic equation of state
is quoted in Sec.~\ref{sec:RET} only for comparison, and no monatomic
numerical curve is included here. In the simulations employing
$m=m_{\rm proton}$, the matter component should therefore be understood as an
effective diatomic test gas rather than as a literal microscopic model of the
post-recombination cosmic medium.

We distinguish two classes of numerical experiments. First, we investigate the evolution of $\Lambda$RET$_6$ with $(\gamma, \bar\Pi)$ chosen in various locations of the hyperbolicity region (see Figure~\ref{dominio}), irrespective of their physical relevance. This is done to (over)emphasize the effects of RET on the cosmic evolution. Secondly, we show that choosing more representative cosmological values of the parameters leads to a rapid convergence of the model to the cosmic dynamics of $\Lambda$CDM. This part of the discussion is performed as a proof-of-concept and is not meant as a complete phenomenological investigation, which is instead deferred to future works. 

\smallskip

Figures~\ref{fig:c1}, \ref{fig:c2} and \ref{fig:c3} show a set of numerical evolutions designed to explore the qualitative role of the $\Lambda$RET$_6$ dynamical pressure in the cosmological dynamics. In these cases the initial data are not chosen with the aim of reproducing a realistic cosmological history. Rather, they are selected inside the hyperbolicity domain discussed in Figure~\ref{dominio}, in order to make the non-equilibrium effects visible on the scale of the plots. This is useful because, for representative post-recombination values of the thermodynamic parameters, the dynamical pressure relaxes very rapidly and the resulting deviations from the corresponding $\Lambda$CDM evolution become extremely small.

Figures~\ref{fig:c1} and Figure~\ref{fig:c2} show the evolution of the scale factor $a$, of the dimensionless inverse temperature parameter $\gamma$, and of the dimensionless dynamical pressure $\bar\Pi$, for initially negative values of $\bar\Pi$, close to the bottom boundary of the hyperbolicity domain. In both cases the dynamical pressure relaxes monotonically towards equilibrium ($\bar\Pi = 0$), with a relaxation rate controlled by the dimensionless relaxation time $\bar\tau$. As expected, smaller values of $\bar\tau$ lead to a faster decay of the non-equilibrium phenomena, while larger values of $\bar\tau$ allow the dynamical pressure to persist longer. The corresponding evolution of the scale factor is much less sensitive to $\bar\tau$, suggesting that the non-equilibrium phenomena impact the expansion history mainly by modifying the evolution of $\gamma$ (it can be argued that the time evolution of the Hubble parameter $H$ is affected by $\bar\Pi$ only indirectly).
\begin{figure} 
	\centering%
	\includegraphics[width=0.45\linewidth]{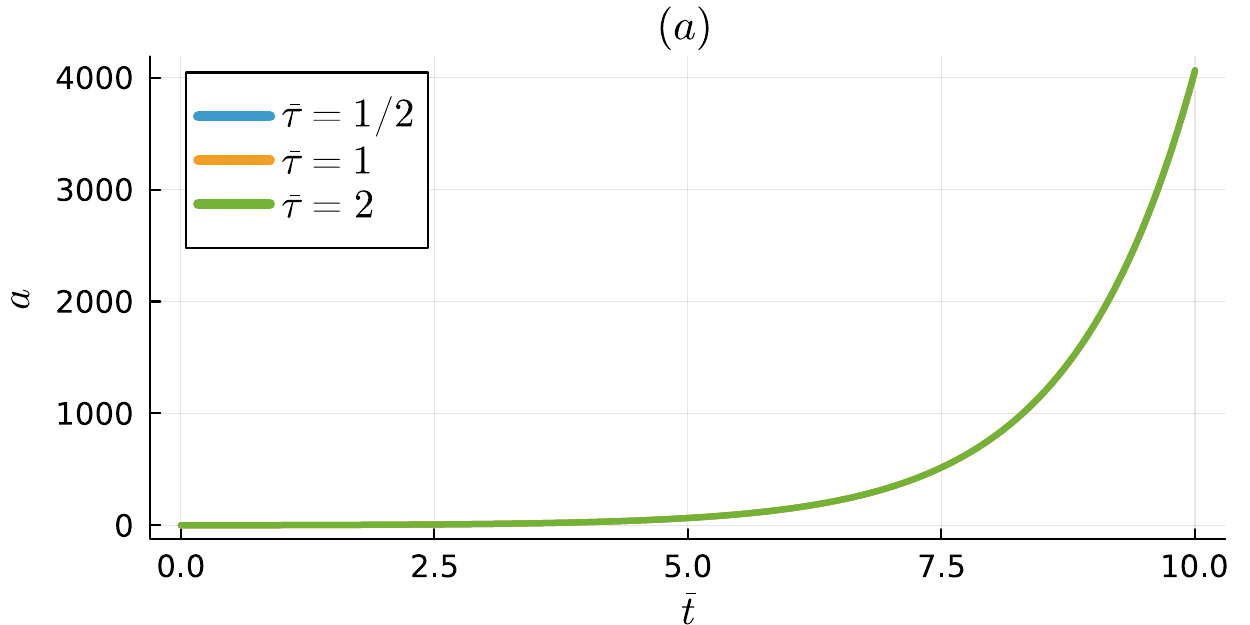}
	\includegraphics[width=0.45\linewidth]{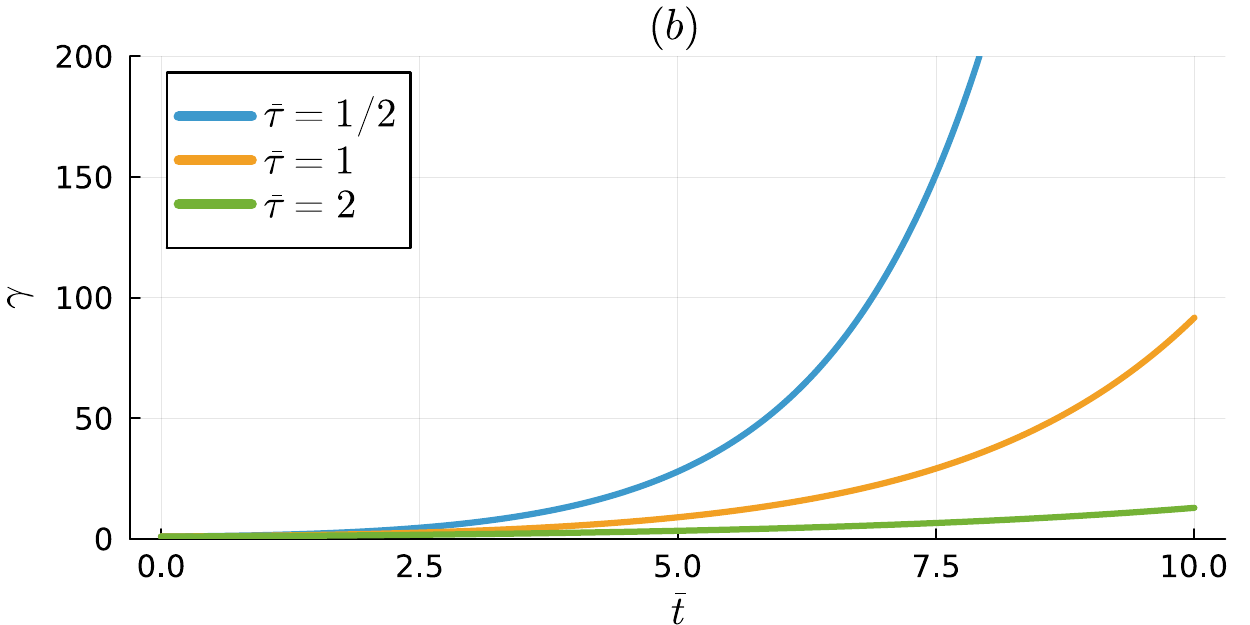}
	\includegraphics[width=0.45\linewidth]{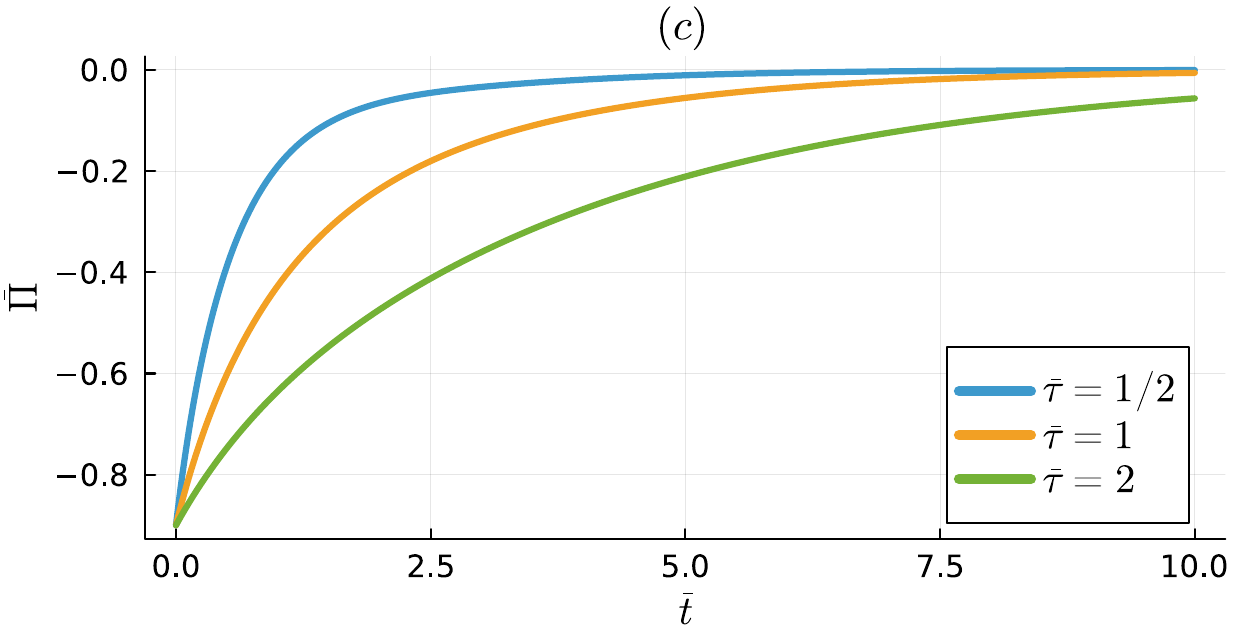}
	\caption{Evolution of $a$, $\gamma$, $\bar\Pi$ with initial conditions $a_0 = 1$, $\gamma_0 = 1$, $\bar\Pi_0 = -0.9$, and three different values of $\bar\tau$ ($\Omega_{M,0} \equiv \Omega_{M,*} = 0.32$, $\Omega_{R,0} \equiv 0$, $H_0 = 67.4$ km s$^{-1}$ Mpc$^{-1}$, $m = m_{proton}$). \label{fig:c1}}
\end{figure}
\begin{figure} 
	\centering%
	\includegraphics[width=0.45\linewidth]{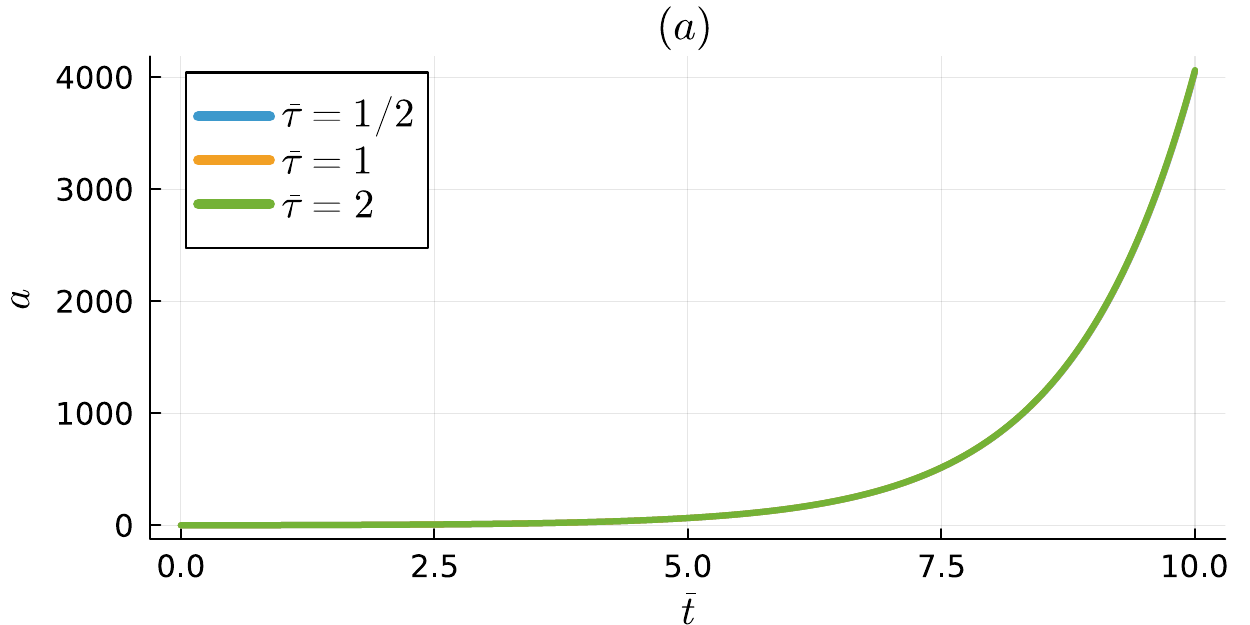}
	\includegraphics[width=0.45\linewidth]{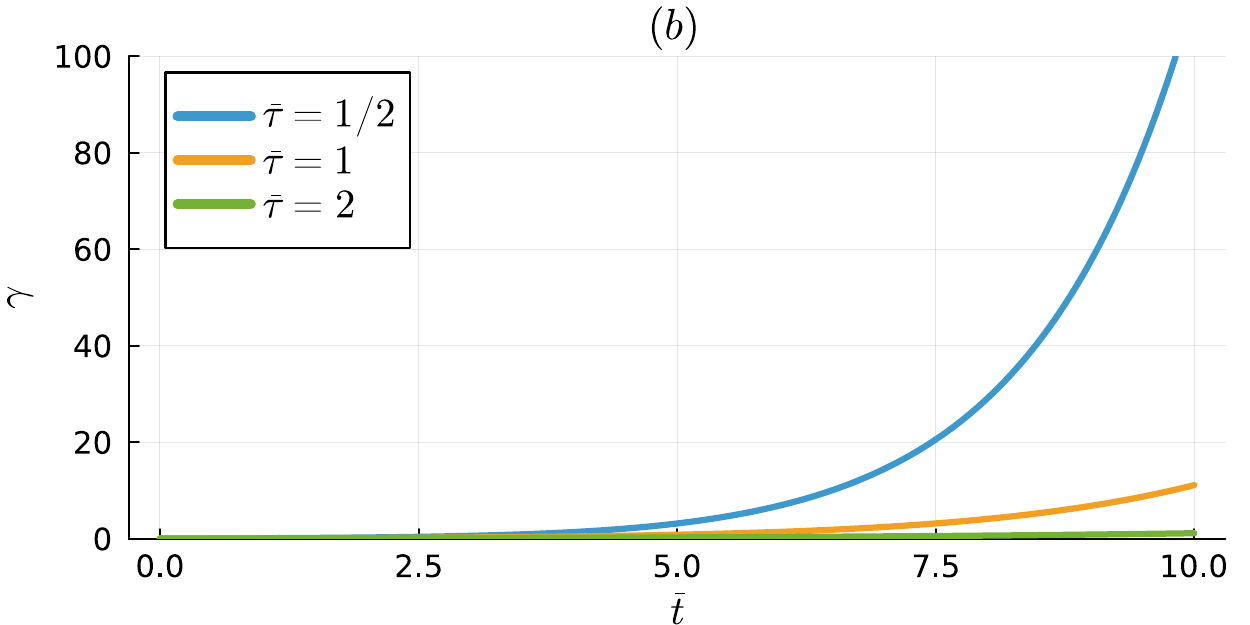}
	\includegraphics[width=0.45\linewidth]{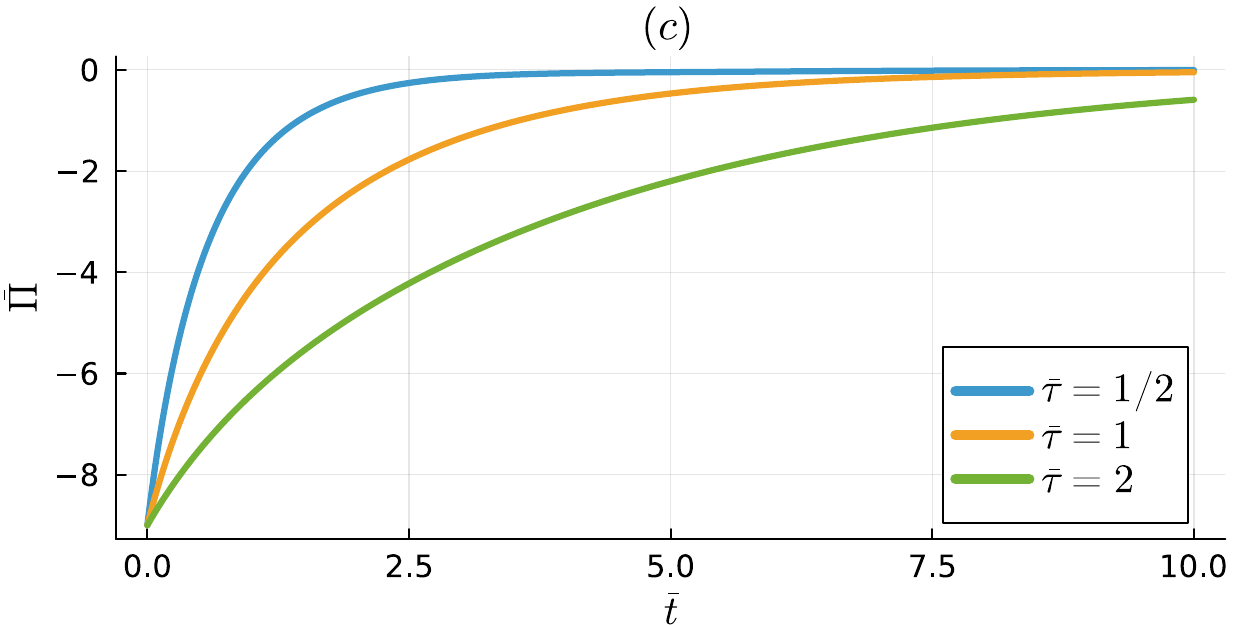}
	\caption{Evolution of $a$, $\gamma$, $\bar\Pi$ with initial conditions $a_0 = 1$, $\gamma_0 = 10^{-1}$, $\bar\Pi_0 = -9$, and three different values of $\bar\tau$ ($\Omega_{M,0} \equiv \Omega_{M,*} = 0.32$, $\Omega_{R,0} \equiv 0$, $m = m_{proton}$). \label{fig:c2}}
\end{figure}

The comparison between Figs.~\ref{fig:c1}(b) and Figure~\ref{fig:c2}(b) also illustrates the dependence on the initial temperature: in Figure~\ref{fig:c1}(b) the initial value is $\gamma_0 = 1$, while in Figure~\ref{fig:c2}(b) one has $\gamma_0 = 10^{-1}$, corresponding to a more relativistic initial state. The qualitative behavior is the same in the two cases: $\bar\Pi$ remains negative and relaxes to zero, while $\gamma$ increases (i.e. the gas cools) as the universe expands. This behavior is shown in Figs.~\ref{fig:c1}(c) and Figure~\ref{fig:c2}(c). The magnitude of the initial non-equilibrium pressure set in the simulation of Figure~\ref{fig:c2} ($\bar\Pi_0 = -9$) is much larger than the one of Figure~\ref{fig:c1} ($\bar\Pi_0 = -0.9$), but in both cases the initial state lies close to the lower boundary of the admissible region depicted in Figure~\ref{dominio}. The evolution of the thermodynamic variables shows a more pronounced transient, even though the scale factor evolution remains extremely close to the corresponding curves of Figure~\ref{fig:c1}, and is only slightly affected by the magnitude of the relaxation time $\bar\tau$.

Figure~\ref{fig:c3} displays an example with initially positive dynamical pressure ($\bar\Pi_0 = 10^{-1}$). In this case the relaxation is not monotone: the dynamical pressure first becomes negative before eventually approaching the equilibrium value. This overshoot, as expected, is more pronounced for larger $\bar\tau$, i.e. for slower relaxation of the non-equilibrium phenomena. 
\begin{figure} 
	\centering%
	\includegraphics[width=0.45\linewidth]{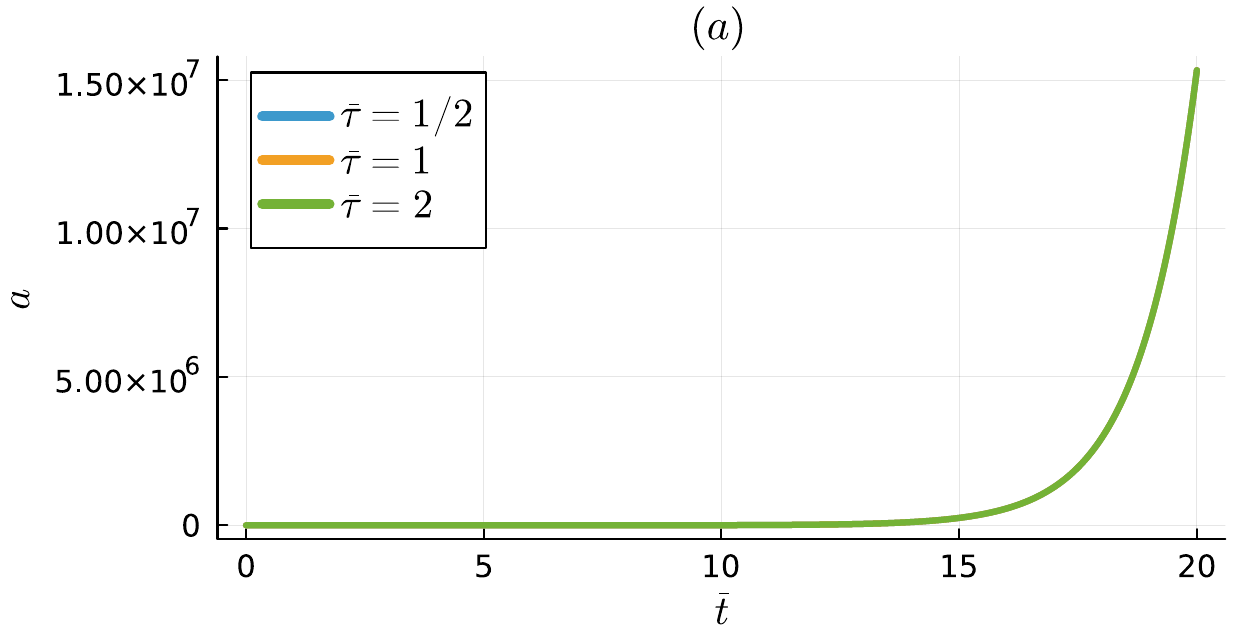}
	\includegraphics[width=0.45\linewidth]{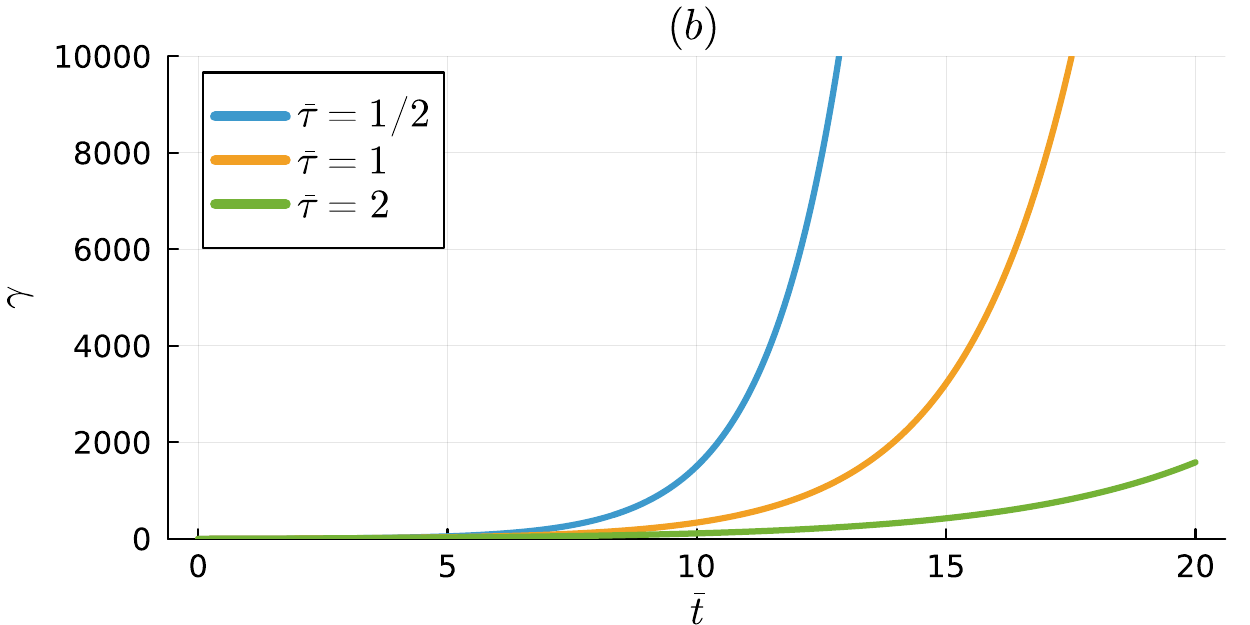}
	\includegraphics[width=0.45\linewidth]{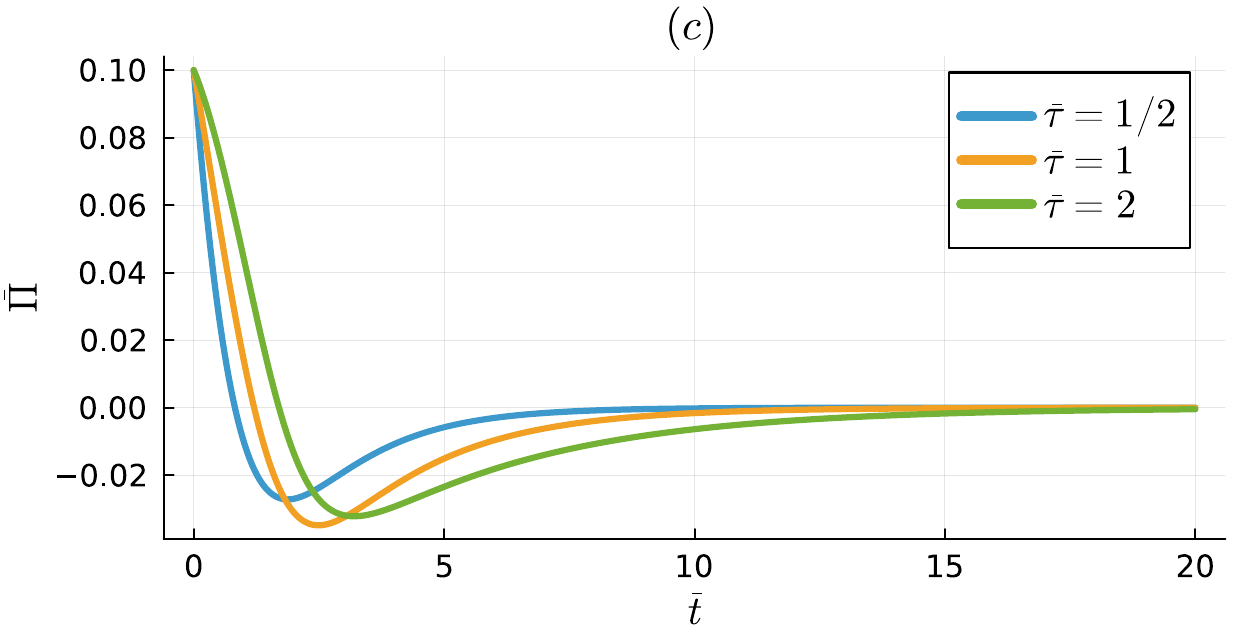}
	\caption{Evolution of $a$, $\gamma$, $\bar\Pi$ with initial conditions $a_0 = 1$, $\gamma_0 = 1$, $\bar\Pi_0 = 0.1$, and three different values of $\bar\tau$ ($\Omega_{M,0} \equiv \Omega_{M,*} = 0.32$, $\Omega_{R,0} \equiv 0$, $m = m_{proton}$). \label{fig:c3}}
\end{figure}

\smallskip 

Figures \ref{fig:c0_a}, \ref{fig:c0_T+Pi}, and \ref{fig:c0_N} show results of the numerical solution of the system with a more physically motivated set of initial data, chosen at the recombination time. The initial scale factor is taken to be
\begin{equation}
	a_0 = \frac{1}{1+z_0}, \qquad z_0 = 1090,	
\end{equation}
with $T_0 =2.7255/a_0 \simeq 2973.5\,{\rm K}$, and with an initially small negative dynamical pressure:
\begin{equation}
	\gamma_0 = \frac{m c^2}{k_B T_0} \simeq 3.7 \times 10^9, \qquad
    \bar{\Pi}_0 = -0.9 / \gamma_0 \approx -2.5 \times 10^{-10}.
\end{equation}
The matter, radiation and cosmological constant density parameters are
normalized at this initial time from the present-day reference values
$\Omega_{M,*}=0.32$ and $\Omega_{R,*}=9.2\times10^{-5}$, giving
$\Omega_{M,0}\simeq0.76$, $\Omega_{R,0}\simeq0.24$, and
$\Omega_{\Lambda,0}\simeq1.2\times10^{-9}$. Here and in the following, the
subscript $0$ refers to the initial recombination epoch, whereas the asterisk
refers to present-day values.

For this comparison we add a separately conserved equilibrium radiation
component, with $p_R=e_R/3$, which is not assigned a dynamical pressure and
couples to the RET$_6$ gas only through the common FLRW expansion rate. The
evolution equations \eqref{eq:ode_a_Lambda}--\eqref{eq:ode_pi_Lambda2}
therefore remain unchanged, with the Hubble function replaced by
\begin{equation}
\mathcal H_{\rm rad}(a,\gamma)
=
\left[
\Omega_{{\rm M},0}
\left(\frac{a_0}{a}\right)^3
\frac{\omega(\gamma)}{\omega_0}
+
\Omega_{{\rm R},0}
\left(\frac{a_0}{a}\right)^4
+
\Omega_{\Lambda,0}
\right]^{1/2}.
\label{eq:Hhat_radiation}
\end{equation}
The initial flatness condition is
$\Omega_{{\rm M},0}+\Omega_{{\rm R},0}+\Omega_{\Lambda,0}=1$.
The density parameters displayed below are obtained by dividing the three
terms in the square brackets in \eqref{eq:Hhat_radiation} by
$\mathcal H_{\rm rad}^2$.

The result is shown in Figure~\ref{fig:c0_a}: the scale factor obtained from $\Lambda$RET$_6$ is practically indistinguishable from the one obtained in the corresponding $\Lambda$CDM model, obtained as the solution of the Friedmann equation:
\begin{equation}
	\dot{a}=aH_0\sqrt{
\Omega_{R,0}\left(\frac{a_0}{a}\right)^4
+\Omega_{M,0}\left(\frac{a_0}{a}\right)^3
+\Omega_{\Lambda,0}},
\end{equation}
where $H_0$ is the Hubble parameter at the initial recombination epoch.
It turns out that the $\Lambda$RET$_6$ correction is rapidly damped and the expansion history matches closely the standard $\Lambda$CDM one (this consideration is largely independent from the chosen value for the relaxation time $\bar\tau$).

The rapid decay of the non-equilibrium phenomena is shown in Figure~\ref{fig:c0_T+Pi}, where the temperature and the dimensionless dynamical pressure are plotted for the same initial data as in Figure~\ref{fig:c0_a}. The temperature decreases rapidly as the universe expands; at the same time, $\bar\Pi$ approaches zero (i.e. the equilibrium value) on a very short timescale compared with the Hubble time. This confirms that, in the regime considered here, the dynamical pressure acts as a quickly damped transient correction that adds up to the dark energy source. In particular, the numerical solution remains close to local equilibrium, which is the regime in which the $\Lambda$RET$_6$ closure and the hyperbolicity analysis are physically meaningful. Note, however, that after decoupling the temperature entering the model is an effective thermodynamic variable and should not be interpreted as a directly measurable temperature.

Finally, Figure~\ref{fig:c0_N} shows the evolution of the density parameters and of the deceleration parameter $q$ as functions of the variable $N=\ln a$. The radiation fraction decreases first, followed by a matter-dominated time interval, and the cosmological constant eventually becomes dominant at late times. The deceleration parameter behaves accordingly: it is positive during the radiation- and matter-dominated eras, and becomes negative once the cosmological constant
dominates. The transition to accelerated expansion is therefore largely driven by $\Lambda$, with the contribution of the non-equilibrium dynamical pressure $\bar\Pi$ of the $\Lambda$RET$_6$ model, which is in agreement with the no-go theorem discussed in Sec.~\ref{sec:nogo}. The role of the dynamical pressure of the $\Lambda$RET$_6$ model is to provide a thermodynamically consistent non-equilibrium correction to the evolution.
\begin{figure} 
	\centering%
	\includegraphics[width=0.6\linewidth]{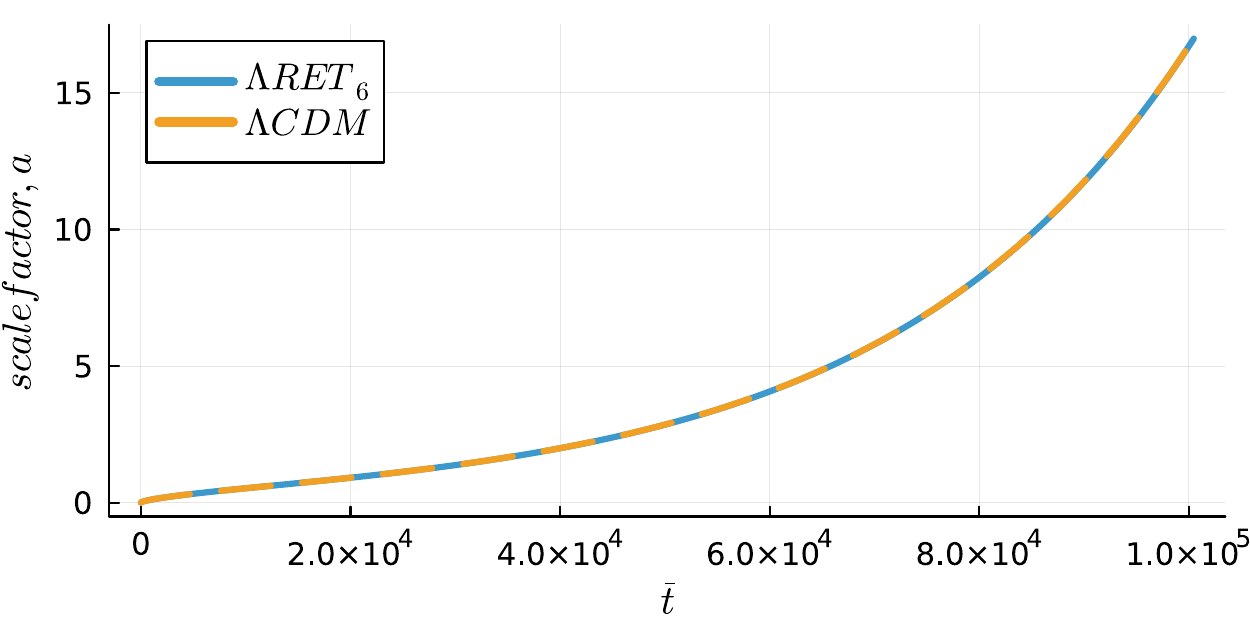}
	\caption{Comparison of the evolution of the scale factor $a$ obtained with the $\Lambda$RET$_6$ and the $\Lambda$CDM models. Initial data: $a_0 = 1/\left(1+z_0\right)$ with $z_0 = 1090$ (recombination era), $T_0 = 2.7255/a_0 \approx 2973.5$ K, $\gamma_0 = mc^2/k_BT_0 \approx 3.7 \times 10^9$, $\bar\Pi_0 = -0.9 / \gamma_0 \approx -2.5 \times 10^{-10}$ ($\bar\tau = 1$; $\Omega_{M,0} = \Omega_{M,*}\,a_0^{-3} / \left(\Omega_{M,*}\,a_0^{-3} + \Omega_{R,*}\,a_0^{-4} + \Omega_{\Lambda,*}\right) \approx 0.76$; $\Omega_{R,0} = \Omega_{R,*}\,a_0^{-4} / \left(\Omega_{M,*}\,a_0^{-3} + \Omega_{R,*}\,a_0^{-4} + \Omega_{\Lambda,*}\right) \approx 0.24$, with $\Omega_{M,*} = 0.32$, $\Omega_{R,*} = 9.2 \times 10^{-5}$; $\Omega_{\Lambda,0} = 1 - \Omega_{M,0} - \Omega_{R,0} \approx 1.2 \times 10^{-9}$; $H_0 = 1.57\times 10^6$ km s$^{-1}$ Mpc$^{-1}$; $m = m_{proton}$). \label{fig:c0_a}}
\end{figure}
\begin{figure} 
	\centering%
	\includegraphics[width=0.45\linewidth]{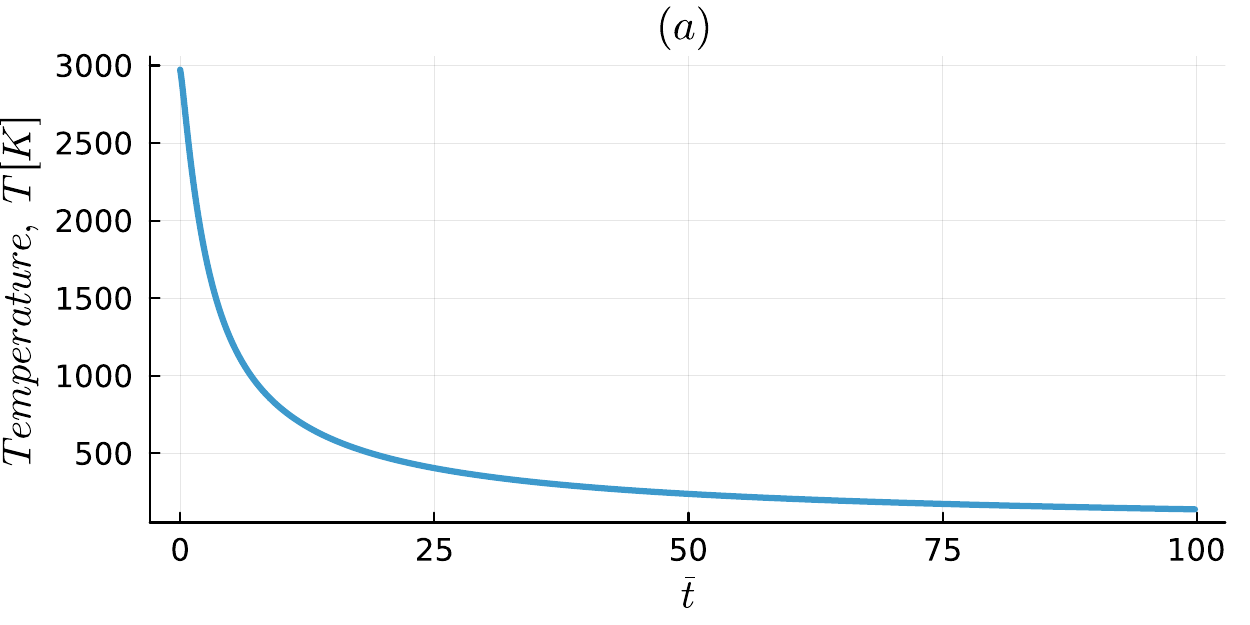}
	\includegraphics[width=0.45\linewidth]{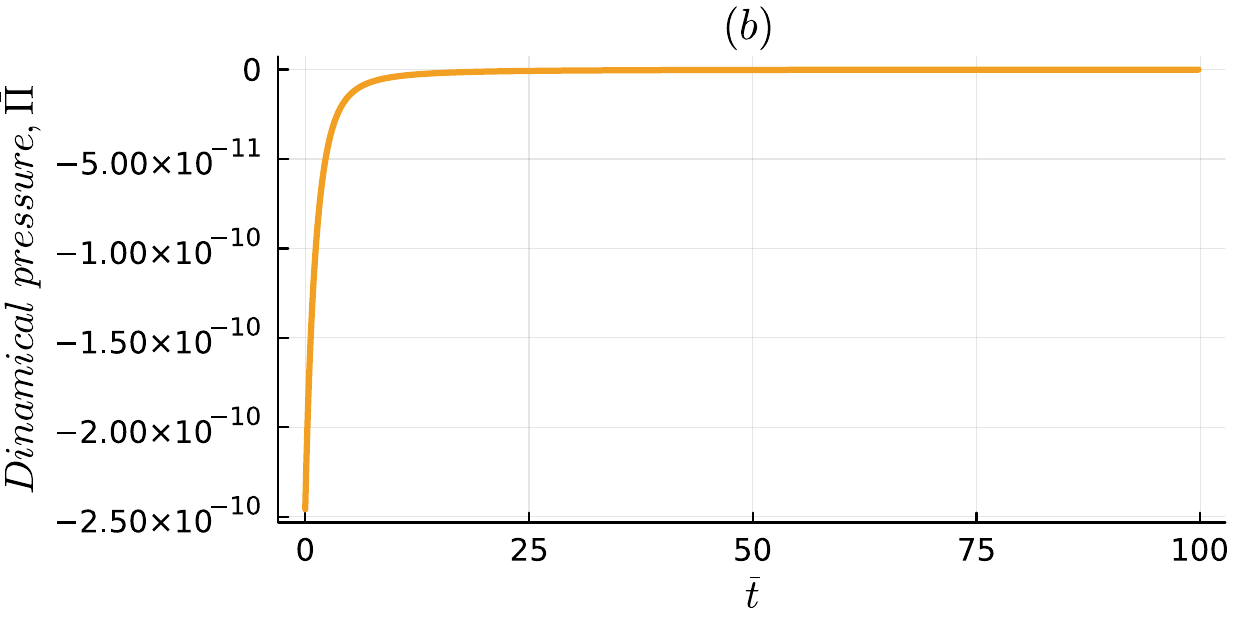}
	\caption{Evolution of the temperature $T$ and the dimensionless dynamical pressure $\bar\Pi$ (same parameters and initial data as in Figure~\ref{fig:c0_a}). \label{fig:c0_T+Pi}}
\end{figure}
\begin{figure} 
	\centering%
	\includegraphics[width=0.6\linewidth]{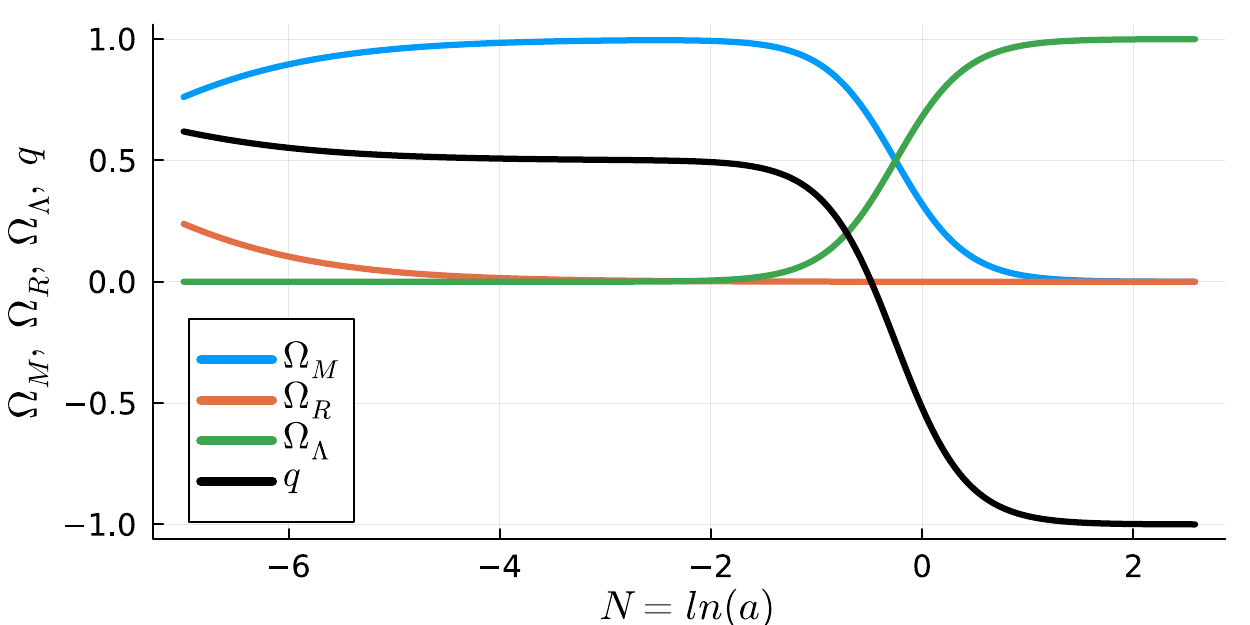}
	\caption{Evolution of the density parameters $\Omega_M = \Omega_{{\rm M},0} \left(\frac{a_0}{a}\right)^3 \left(\frac{\omega(\gamma)}{\omega_0}\right) \mathcal H_{\rm rad}^{-2}$, $\Omega_R = \Omega_{{\rm R},0} \left(\frac{a_0}{a}\right)^4 \mathcal H_{\rm rad}^{-2}$, $\Omega_\Lambda = \Omega_{\Lambda,0} \mathcal H_{\rm rad}^{-2}$ and of the deceleration parameter $q$ obtained with the $\Lambda$RET$_6$ model. Initial data: $a_0 = 1/\left(1+z_0\right)$ with $z_0 = 1090$ (recombination era), $T_0 = 2.7255/a_0 \approx 2973.5$ K, $\gamma_0 = mc^2/k_BT_0 \approx 3.7 \times 10^9$, $\bar\Pi_0 = -0.9 / \gamma_0 \approx -2.5 \times 10^{-10}$ ($\bar\tau = 1$; $\Omega_{M,0} = \Omega_{M,*}\,a_0^{-3} / \left(\Omega_{M,*}\,a_0^{-3} + \Omega_{R,*}\,a_0^{-4} + \Omega_{\Lambda,*}\right) \approx 0.76$; $\Omega_{R,0} = \Omega_{R,*}\,a_0^{-4} / \left(\Omega_{M,*}\,a_0^{-3} + \Omega_{R,*}\,a_0^{-4} + \Omega_{\Lambda,*}\right) \approx 0.24$, with $\Omega_{M,*} = 0.32$, $\Omega_{R,*} = 9.2 \times 10^{-5}$; $\Omega_{\Lambda,0} = 1 - \Omega_{M,0} - \Omega_{R,0} \approx 1.2 \times 10^{-9}$; $m = m_{proton}$). \label{fig:c0_N}}
\end{figure}
Taken together, these simulations show two complementary aspects of the model. When the initial data are chosen near the boundary of the admissible hyperbolicity region, $\Lambda$RET$_6$ effects can produce small, albeit appreciable, transient modifications in the evolution of the thermodynamic quantities and in the relaxation of the dynamical pressure. When representative post-recombination initial data are used, the same non-equilibrium mode decays rapidly and the cosmic expansion converges to the $\Lambda$CDM one, as expected for the parameter range considered here. 

For the parameter values considered here, the background expansion history
of the $\Lambda$RET$_6$ model becomes practically indistinguishable from that
of the corresponding $\Lambda$CDM model, while retaining a thermodynamically
consistent kinetic-theory mechanism for small non-equilibrium corrections. A
complete assessment of its phenomenological viability, including cosmological
perturbations and observational constraints, is left for future work.

%
%
%
\section{Conclusions}
\label{sec:conclusions}
In this work we have presented a general-relativistic model of non-equilibrium gases on curved spacetimes within the framework of Rational Extended Thermodynamics. 

Notably, in this approach to dissipative relativistic fluids the constitutive
structure is derived from the underlying relativistic kinetic theory of
  polyatomic gases through the Maximum Entropy Principle and a
generalized Boltzmann--Chernikov kinetic equation. The explicit model is fixed
once  
the relaxation time, have been specified.

We focused specifically on RET$_6$, since the only non-equilibrium variable is the
dynamical pressure $\Pi$, promoting its Minkowski space formulation to a general curved spacetime by means of the {\em minimal coupling principle}.

The first notable result we found is a {\em kinetic-theory no-go theorem}, which states that for any admissible one-particle distribution, the RET$_6$ stress-energy tensor automatically satisfies the strong energy condition. This result can be traced back to the kinetic closure of the considered theory.

Then, we specialized our investigation to the study of the RET$_6$ gas as a source for a homogeneous and isotropic FLRW spacetime as a practical implementation of the general theory. In particular, we have shown how the non-equilibrium dynamical pressure modifies the expansion dynamics with respect to the perfect-fluid (Euler) case. Consistently with the kinetic-theory no-go theorem, cosmic acceleration cannot be generated by the RET$_6$ gas alone. We then reintroduced the cosmological constant term, formulating what we dubbed as $\Lambda$RET$_6$ model, i.e. the RET$_6$ gas plus the cosmological constant. For the diatomic equation of state and a constant positive relaxation time, we established the existence and local stability of a de~Sitter attractor at late times.

Finally, we performed a numerical investigation of $\Lambda$RET$_6$, both with and without the inclusion of radiation, showing that the expansion history rapidly converges to that of $\Lambda$CDM, while exhibiting small deviations gauged by the relaxation time and the initial non-equilibrium state of the gas. Thus, the role of RET$_6$ is not to replace dark energy, but to provide a thermodynamically consistent kinetic-theory mechanism for transient non-equilibrium corrections to the standard cosmological expansion.

\begin{acknowledgments}
    \noindent A.G. is supported by the Italian Ministry of Universities and Research (MUR) through the grant ``BACHQ: Black Holes and The Quantum'' (grant no. J33C24003220006). A.G. and A.M. are partially supported by the INFN grant FLAG. This work has been carried out in the framework of activities of the National Group of Mathematical Physics (GNFM, INdAM).
\end{acknowledgments}

\section*{Data availability}
\noindent No external datasets were used in this work. The numerical data underlying the figures are generated from the equations and parameter values reported in the
manuscript and are available from the corresponding author upon reasonable
request.

\section*{Conflict of interest}
\noindent The authors declare no competing interests.

\section*{Author contributions}
\noindent All authors contributed equally to the conception, development, and writing of the manuscript.

%
%
%
\appendix
\section{Characteristic Velocities}
\label{app:char_vel}
We begin with the \emph{complete} RET$_6$ system in flat spacetime,
including the production term in the balance law for the dynamical pressure:
\begin{subequations}
\label{eq:RET6_system}
\begin{align}
    &\partial_{\alpha}V^{\alpha}=0,
    \label{eq:sys_mass}
    \\
    &\partial_{\alpha}T^{\alpha\beta}=0,
    \label{eq:sys_momentum}
    \\
    &u^{\alpha}\partial_{\alpha}\Pi
    +c^2\rho\bigl(L_2+L_1\bar{\Pi}\bigr)
       u^{\alpha}\partial_{\alpha}\gamma
    +\Pi\partial_{\alpha}u^{\alpha}
    =-\frac{\Pi}{\tau},
    \label{eq:sys_Pi}
    \\
    &p=\frac{\rho c^2}{\gamma},
    \qquad
    e=\rho c^2\omega(\gamma).
    \label{eq:sys_eos}
\end{align}
\end{subequations}

The source term in Eq.~\eqref{eq:sys_Pi} is essential for the relaxation
dynamics. Since it contains no derivatives of the field variables, however,
it does not contribute to the characteristic polynomial. Therefore,
\emph{only for the calculation of the characteristic velocities}, we set
the right-hand side of Eq.~\eqref{eq:sys_Pi} equal to zero and apply the
standard characteristic substitution \cite{JMP}
\begin{equation}
\label{chain}
    \partial_{\alpha}
    \longrightarrow
    \left(
        \nu_{\alpha}
        -\frac{\lambda}{c}u_{\alpha}
    \right)\delta .
\end{equation}
Here $\lambda$ is the dimensionless characteristic velocity, measured in
units of $c$ in the local rest frame, while $\delta$ denotes the amplitude
of the discontinuity in the first derivatives. The spacelike unit covector
$\nu_{\alpha}$ satisfies
\begin{equation}
    \nu_{\alpha}\nu^{\alpha}=-1,
    \qquad
    \nu_{\alpha}u^{\alpha}=0.
\end{equation}
Applying Eq.~\eqref{chain} to particle conservation,
$\partial_{\alpha}(\rho u^{\alpha})=0$, gives
\begin{equation}
    \left(\nu_{\alpha}-\frac{\lambda}{c}u_{\alpha}\right)
\left(u^{\alpha}\delta\rho+\rho\,\delta u^{\alpha}\right)=0.
\end{equation}
Since $\nu_{\alpha}u^{\alpha}=0$ and
$u_{\alpha}\delta u^{\alpha}=0$, the latter relation becomes
\begin{equation}
\label{eq:vdu}
    \nu_{\alpha}\delta u^{\alpha}
    =\frac{\lambda c}{\rho}\,\delta\rho.
\end{equation}
For the energy-momentum tensor we have
\begin{equation}
\begin{split}
    \delta T^{\alpha\beta}
    ={}&\frac{\delta(e+p+\Pi)}{c^2}u^{\alpha}u^{\beta}
    +\frac{e+p+\Pi}{c^2}
       \left(u^{\alpha}\delta u^{\beta}
       +\delta u^{\alpha}u^{\beta}\right)
    \\
    &-\delta(p+\Pi)g^{\alpha\beta}.
\end{split}
\end{equation}
The characteristic form of energy--momentum conservation is therefore
\begin{equation}
\label{eq:char_T}
    \left(\nu_{\alpha}-\frac{\lambda}{c}u_{\alpha}\right)
    \delta T^{\alpha\beta}=0.
\end{equation}
Projection of Eq.~\eqref{eq:char_T} along $u_{\beta}$ gives the energy
relation
\begin{equation}
\label{eq:energysector}
    -\lambda c\,\delta e
    +(e+p+\Pi)\nu_{\alpha}\delta u^{\alpha}=0,
\end{equation}
whereas projection onto the local rest space gives the momentum relation
\begin{equation}
\label{eq:momentumsector}
    -\frac{\lambda}{c}(e+p+\Pi)
       \delta u^{\beta}h_{\beta}{}^{\gamma}
    +\delta(p+\Pi)\nu^{\gamma}=0.
\end{equation}
Using Eq.~\eqref{eq:vdu} in Eq.~\eqref{eq:energysector}, we obtain
\begin{equation}
\label{eq:de}
    \delta e=\frac{e+p+\Pi}{\rho}\,\delta\rho.
\end{equation}
On the other hand, the equation of state
$e=\rho c^2\omega(\gamma)$ yields
\begin{equation}
\label{eq:de_thermo}
    \delta e
    =c^2\omega\,\delta\rho
    +\rho c^2\omega'(\gamma)\,\delta\gamma.
\end{equation}
Equating Eqs.~\eqref{eq:de} and \eqref{eq:de_thermo}, and using
$e=\rho c^2\omega$, gives
\begin{equation}
\label{eq:dgamma}
    \delta\gamma
    =\frac{p+\Pi}{\rho^2c^2\omega'}\,\delta\rho.
\end{equation}
Contracting Eq.~\eqref{eq:momentumsector} with $\nu_{\gamma}$ and using the
positive spatial projector
$h^{\alpha\beta}=u^{\alpha}u^{\beta}/c^2-g^{\alpha\beta}$, we find
\begin{equation}
    \frac{\lambda}{c}(e+p+\Pi)
       \delta u^{\beta}\nu_{\beta}
    -\delta(p+\Pi)=0.
\end{equation}
Together with Eq.~\eqref{eq:vdu}, this gives
\begin{equation}
\label{eq:dpP}
    \delta(p+\Pi)
    =\lambda^2(e+p+\Pi)\frac{\delta\rho}{\rho}.
\end{equation}
Applying the characteristic substitution to the differential part of the
left-hand side of Eq.~\eqref{eq:sys_Pi} gives
\begin{equation}
    -\lambda c\,\delta\Pi
    -\lambda c^3\rho\bigl(L_2+L_1\bar{\Pi}\bigr)\delta\gamma
    +\Pi\nu_{\alpha}\delta u^{\alpha}=0.
\end{equation}
For a nonzero characteristic root, division by $-\lambda c$ and use of
Eq.~\eqref{eq:vdu} lead to
\begin{equation}
\label{eq:dpi}
    \delta\Pi
    =-c^2\rho\bigl(L_2+L_1\bar{\Pi}\bigr)\delta\gamma
    +\Pi\frac{\delta\rho}{\rho}.
\end{equation}
For the thermodynamic pressure, from
$p=\rho c^2/\gamma$ we get
\begin{equation}
\label{eq:dp_thermo}
    \delta p
    =\frac{p}{\rho}\,\delta\rho
    -\frac{c^2\rho}{\gamma^2}\,\delta\gamma.
\end{equation}
Adding Eqs.~\eqref{eq:dpi} and \eqref{eq:dp_thermo}, substituting
Eq.~\eqref{eq:dgamma}, and comparing with Eq.~\eqref{eq:dpP}, we obtain
\begin{equation}
    \lambda^2(e+p+\Pi)
    =(p+\Pi)
    \left[
        1-\frac{1}{\gamma^2\omega'}
        -\frac{1}{\omega'}\bigl(L_2+L_1\bar{\Pi}\bigr)
    \right].
\end{equation}
Finally, using
\begin{equation}
    k=1-\frac{1}{\gamma^2\omega'},
    \qquad
    \hat c_V=-\gamma^2\omega',
\end{equation}
we find the two nonzero longitudinal characteristic velocities
\begin{equation}
\label{eq:lambda_pm}
    \lambda_{\pm}
    =\pm\sqrt{
        \frac{p+\Pi}{e+p+\Pi}
        \left[
            k+\frac{\gamma^2}{\hat c_V}
            \bigl(L_2+L_1\bar{\Pi}\bigr)
        \right]
    }.
\end{equation}
Equivalently, their common squared magnitude is
\begin{equation}
    \lambda^2
    =\frac{p+\Pi}{e+p+\Pi}
    \left[
        k+\frac{\gamma^2}{\hat c_V}
        \bigl(L_2+L_1\bar{\Pi}\bigr)
    \right].
\end{equation}
For completeness, the division by $\lambda$ made above isolates the
propagating longitudinal modes. Setting $\lambda=0$ directly in the principal
system gives four non-propagating material characteristic modes. Two of them
are transverse velocity perturbations satisfying
\[
    \nu_\alpha\delta u^\alpha=0,
\]
in addition to the normalization constraint
$u_\alpha\delta u^\alpha=0$. The remaining two are independent scalar
perturbations in $(\delta\rho,\delta\gamma,\delta\Pi)$ subject to the single
constraint
\[
    \delta(p+\Pi)=0.
\]
Hence, in the local rest frame, the complete characteristic spectrum consists
of
\begin{equation}
    \lambda=0 \quad \text{(multiplicity four)},
    \qquad
    \lambda=\lambda_{-},\lambda_{+}.
\end{equation}
The strict inequalities $0<\lambda_{+}^{2}<1$ are therefore precisely the
conditions ensuring hyperbolicity of the propagating sector and causality in
the local rest frame.
\section{Jacobian Matrix}
\label{app:jacobian}
In this Appendix, we explicitly derive the components of the Jacobian matrix evaluated at the de Sitter fixed point. For brevity, the evaluation of any quantity at this origin (setting $x=0$, $\xi=0$, and $\bar{\Pi}=0$) is denoted throughout by the subscript $0$, i.e., $|_0$. 

The late-time stability analysis requires expanding the right-hand sides of the dynamical system around the origin of the phase space. 
First, we note that as the scale factor diverges ($x \to 0$), the matter density $\rho=\rho_0(a_0x)^3$ vanishes. Therefore, the dynamical Hubble parameter can be expanded as:
\be
H(x,\xi) = \sqrt{\frac{8\pi G \rho_0 a_0^3}{3} x^3
\omega(\xi) + \frac{\Lambda c^2}{3}}
= H_{\text{dS}} + \mathcal{O}(x^3) \ .
\ee
This immediately yields the decoupled equation for the spatial sector at linear order:
\be
\dot{x} \simeq -H_{\text{dS}} \, x \quad \implies \quad \left. \frac{\partial \dot{x}}{\partial x} \right|_0 = -H_{\text{dS}} \ ,
\ee
which provides the $J_{00}$ component and the zeros in the first row and column of the Jacobian.

To regularize the thermodynamic sector, we perform an asymptotic expansion of the RET6 coefficients in the low-temperature limit, $\gamma \to \infty$ (i.e., $\xi \to 0$). For the diatomic equation of state considered in the stability analysis, the
low-temperature expansion yields:
\be
\omega(\xi) \approx 1 + \frac{5}{2}\xi + \mathcal{O}(\xi^2) \quad \implies \quad \left. \omega_{,\xi} \right|_0 = \frac{5}{2} \ .
\ee
Substituting this into the prefactor of the thermodynamic equations and evaluating it at the fixed point, we find:
\be
\left. \frac{3 H(x,\xi)}{\omega_{,\xi}} \right|_0 = \frac{6}{5} H_{\text{dS}} \ .
\ee
This allows us to linearize the evolution equation for $\xi$:
\be
\dot{\xi} \simeq -\frac{6}{5} H_{\text{dS}} (\bar{\Pi} + \xi) \ ,
\ee
from which we trivially extract the second row of the Jacobian:
\be
\left. \frac{\partial \dot{\xi}}{\partial \xi} \right|_0 = -\frac{6}{5} H_{\text{dS}} \ , \quad \quad \left. \frac{\partial \dot{\xi}}{\partial \bar{\Pi}} \right|_0 = -\frac{6}{5} H_{\text{dS}} \ .
\ee

Finally, we address the evolution equation for the dynamical pressure $\bar{\Pi}$. The presence of the $\xi^2$ term in the denominator of $\dot{\bar{\Pi}}$ is regularized by the asymptotic behavior of the coefficients $L_1(\xi)$ and $L_2(\xi)$. In the classical limit, the Taylor expansion of the ratio for diatomic gases yields:
\be
\lim_{\xi \to 0} \frac{L_1(\xi)}{\xi^2} = -\frac{9}{2} \ , \quad \quad \lim_{\xi \to 0} \frac{L_2(\xi)}{\xi^2} = \frac{2}{3} \ .
\ee
Consequently, one has:
\be
\frac{L_1 \bar{\Pi} + L_2}{\xi^2} \simeq -\frac{9}{2} \bar{\Pi} + \frac{2}{3} \ .
\ee
Substituting this expansion into the exact $\dot{\bar{\Pi}}$ equation, we obtain the linearized form:
\be
\dot{\bar{\Pi}} \simeq -\frac{6}{5} H_{\text{dS}} (\bar{\Pi} + \xi) \left( -\frac{9}{2} \bar{\Pi} + \frac{2}{3} \right) - \frac{\bar{\Pi}}{\tau} \ .
\ee
Evaluating the partial derivatives at the fixed point, we are left with:
\begin{align}
\left. \frac{\partial \dot{\bar{\Pi}}}{\partial \xi} \right|_0 &= -\frac{6}{5} H_{\text{dS}} \left(\frac{2}{3}\right) = -\frac{4}{5} H_{\text{dS}} \ , \\
\left. \frac{\partial \dot{\bar{\Pi}}}{\partial \bar{\Pi}} \right|_0 &= -\frac{6}{5} H_{\text{dS}} \left(\frac{2}{3}\right) - \frac{1}{\tau} = -\frac{4}{5} H_{\text{dS}} - \frac{1}{\tau} \ .
\end{align}
These derivatives constitute the third row of the Jacobian matrix.

\bibliographystyle{apsrev4-2}
\bibliography{Bib_GRET_submission}

\end{document}